\documentclass[11pt]{article}
\usepackage{fullpage}
\usepackage{amsmath,amsfonts,amsthm,amssymb}
\usepackage{url}
\usepackage[usenames,dvipsnames,svgnames,table]{xcolor}
\usepackage[pdftex,bookmarks,colorlinks=true,
            linkcolor=red,
            urlcolor=blue,
            citecolor=gray]{hyperref}
\usepackage{color}
\usepackage{algorithm,algorithmic}
\usepackage{dsfont}
\newtheorem{theorem}{Theorem}
\newtheorem{lemma}[theorem]{Lemma}
\newtheorem{corollary}[theorem]{Corollary}

\newtheorem{definition}[theorem]{Definition}
\newcommand{\specialcell}[2][c]{%
  \begin{tabular}[#1]{@{}c@{}}#2\end{tabular}}

  \usepackage{nth}
  \usepackage{intcalc}

\newcommand{\R}{\mathbb{R}}

\newcommand{\E}{\mathbb{E}}
\let\Pr\relax
\DeclareMathOperator*{\Pr}{\mathbb{P}}
\newcommand{\variance}{\mathrm{Var}}

\newcommand{\norm}[1]{\|#1\|}
\newcommand{\normTwo}[1]{\norm{#1}_{2}}
\newcommand{\normFull}[1]{\left\|#1\right\|}
\newcommand{\normFro}[1]{\norm{#1}_{\mathrm{F}}}

\newcommand{\abs}[1]{\lvert#1\rvert}
\newcommand{\inprod}[1]{\left\langle #1 \right\rangle}

\newcommand{\gradient}{\bigtriangledown}
\newcommand{\grad}{\gradient}
\newcommand{\hessian}{\gradient^{2}}
\newcommand{\hess}{\hessian}

\DeclareMathOperator*{\argmin}{arg\,min}

\DeclareMathOperator{\rank}{rank}
\DeclareMathOperator{\nnz}{nnz}

\newcommand{\bv}[1]{\mathbf{#1}}

\newcommand{\wt}{\widetilde}

\newcommand{\poly}{\mathop\mathrm{poly}}
\newcommand{\otilde}{\widetilde{O}}

\newcommand{\ma}{\mvar{a}}
\newcommand{\mb}{\mvar{b}}

\newcommand{\mSigma}{\mvar{\Sigma}}

\newcommand{\eqdef}{\mathbin{\stackrel{\rm def}{=}}}
\newcommand{\dist}{\mathcal{D}}
\newcommand{\supp}{\mathrm{supp}}

\newcommand{\sk}[1]{\noindent{\textcolor{cyan}{{\bf sk:} \em #1}}}

\newcommand{\gap}{\mathrm{gap}}
\DeclareMathOperator{\nrank}{sr}
\DeclareMathOperator{\nvar}{v}
\newcommand{\rayquot}{\mathrm{quot}}
\newcommand{\ssvrgstep}{\mathrm{ssvrg\_iter}}

\newcommand{\opt}[1]{{#1}^{\mathrm{opt}}}

\makeatletter 

\global\long\def\boldVar#1{\mathbf{#1}}
\global\long\def\mvar#1{\boldVar{#1}}

\global\long\def\defeq{\stackrel{\mathrm{{\scriptscriptstyle def}}}{=}}

 \renewcommand{\norm}[1]{\left\|#1\right\|}

\global\long\def\ma{\mvar A}
 \global\long\def\mb{\mvar B}

\global\long\def\mI{\mvar I}

\global\long\def\abs#1{\left|#1\right|}

\global\long\def\ceil#1{\left\lceil #1 \right\rceil }

\global\long\def\lamtilij#1#2{\widetilde{\lambda}_{#2}^{\left(#1\right)}}

\global\long\def\lamhatij#1#2{\widehat{\lambda}_{#2}^{\left(#1\right)}}

\global\long\def\lamj#1{\lambda_{#1}}

\global\long\def\lambari#1{\overline{\lambda}^{\left(#1\right)}}

\newcommand{\expec}[1]{\mathbb{E}\left[#1\right]}
\newcommand{\prob}[1]{\mathbb{P}\left[#1\right]}

\newcommand{\xtilde}{\widetilde{x}}

\newcommand{\lamiBinv}[1]{\lambda_{#1}\left(\mb^{-1}\right)}
\newcommand{\lambdah}{\widehat{\lambda}}

\newcommand{\xhat}{\widehat{x}}
\newcommand{\mcalF}{\mathcal{F}}
\newcommand{\mcalG}{\mathcal{G}}
\newcommand{\rayquoth}[1]{\widehat{\mathrm{quot}}\left(#1\right)}
\newcommand{\solve}[1]{\mathrm{solve}\left(#1\right)}

\newcommand{\set}[1]{\left\{#1\right\}}

\newcommand{\mM}{\mathbf{M}}

\begin{document}

\title{Faster Eigenvector Computation via \\Shift-and-Invert Preconditioning~\footnote{This paper combines work first appearing in \cite{garber2015fast} and \cite{jin2015robust}}}

\date{}
\author{Dan Garber \\ Toyota Technological Institute at Chicago \\ \texttt{dgarber@ttic.edu}
\and
Elad Hazan \\ Princeton University \\ \texttt{ehazan@cs.princeton.edu}
\and
Chi Jin \\ UC Berkeley \\ \texttt{chijin@eecs.berkeley.edu}
\and
Sham M. Kakade \\ University of Washington \\ \texttt{sham@cs.washington.edu}
\and
Cameron Musco\\MIT\\ \texttt{cnmusco@mit.edu}
\and
Praneeth Netrapalli \\ Microsoft Research, New England \\ \texttt{praneeth@microsoft.com}
\and
Aaron Sidford \\ Microsoft Research, New England \\ \texttt{asid@microsoft.com} }
\maketitle

\begin{abstract}
	We give faster algorithms and improved sample complexities for estimating the top eigenvector of a matrix $\bv{\Sigma}$ -- i.e. computing a unit vector $x$ such that $x^\top \bv{\Sigma}x \ge (1-\epsilon)\lambda_1(\bv{\Sigma})$:
	
	\begin{itemize}
		\item \textbf{Offline Eigenvector Estimation:}  Given an explicit $\bv{A} \in \R^{n \times d}$ with $\bv{\Sigma} = \bv{A}^\top \bv{A}$, we show how to compute an $\epsilon$ approximate top eigenvector in time $\otilde \left(\left [\nnz(\bv{A}) + \frac{d\nrank(\bv{A})}{\gap^2}\right ]\cdot \log 1/\epsilon \right )$ and $\otilde\left(\left [\frac{\nnz(\bv A)^{3/4} (d\nrank(\bv{A}))^{1/4}}{\sqrt{\gap}}\right ]\cdot \log1/\epsilon \right )$. Here $\nnz(\bv{A})$ is the number of nonzeros in $\bv{A}$, $\nrank(\bv{A}) \eqdef \frac{\normFro{\bv{A}}^2}{\normTwo{\bv{A}}^2}$ is the stable rank, $\gap$ is the relative eigengap, and $\tilde O(\cdot)$ hides log factors in $d$ and $\gap$. By separating the $\gap$ dependence from the $\nnz(\bv{A})$ term, our first runtime improves upon the classical power and Lanczos methods. It also improves prior work using fast subspace embeddings \cite{ailon2009fast,clarkson2013low} and stochastic optimization \cite{shamir2015stochastic}, giving significantly better dependencies on $\nrank(\bv{A})$ and $\epsilon$. Our second running time improves these further when $\nnz(\bv{A}) \le \frac{d\nrank(\bv{A})}{\gap^2}$.
		
		\item \textbf{Online Eigenvector Estimation:} 
		Given a distribution $\dist$ with covariance matrix $\bv{\Sigma}$ and a vector $x_0$ which is an $O(\gap)$ approximate top eigenvector for $\bv{\Sigma}$, we show how to
		refine to an $\epsilon$ approximation using $ O \left (\frac{\nvar(\dist)}{\gap \cdot \epsilon} \right )$ samples from $\dist$.  Here $\nvar(\dist)$ is a natural notion of variance. 
		Combining our algorithm with previous work to initialize $x_0$, we obtain improved sample complexity and runtime results under a variety of assumptions on $\dist$. 
	\end{itemize}
	
	We achieve our results using a general framework that we believe is of independent interest. We give a robust analysis of the classic method of \emph{shift-and-invert} preconditioning to reduce eigenvector computation to \emph{approximately} solving a sequence of linear systems. We then apply fast stochastic variance reduced gradient (SVRG) based system solvers to achieve our claims. We believe our results suggest the general effectiveness of shift-and-invert based approaches and imply that further computational gains may be reaped in practice.
\end{abstract}

\thispagestyle{empty}
\clearpage
\setcounter{page}{1}

\section{Introduction}

Given $\bv{A} \in \mathbb{R}^{n \times d}$, computing the top eigenvector
of $\bv{A}^{\top}\bv{A}$ is a fundamental problem in numerical linear algebra, applicable to principal component analysis \cite{jolliffe2002principal}, spectral clustering and learning \cite{ng2002spectral,vempala2004spectral}, pagerank computation, and many other graph computations \cite{page1999pagerank,koren2003spectral,spielman2007spectral}. For instance, a degree-$k$ principal component analysis is nothing more than performing $k$ leading eigenvector computations. 
Given the ever-growing size of modern datasets, it is thus a key challenge to come up with more efficient algorithms for this basic computational primitive.  

In this work we provide improved algorithms for computing the top eigenvector,
both in the \emph{offline} case, when $\bv{A}$ is given
explicitly and in the \emph{online} or \emph{statistical} case where we access samples from a distribution $\mathcal{D}$ over
$\mathbb{R}^{d}$ and wish to estimate the top eigenvector of the covariance matrix $\E_{a\sim\dist} \left [aa^\top \right ]$.
In the offline case, our algorithms are the fastest to date in a wide and meaningful regime of parameters. Notably, while the running time of most popular methods for eigenvector computations is a product of the size of the dataset (i.e. number of non-zeros in $\bv{A}$) and certain spectral characteristics of $\bv{A}$, which both can be quite large in practice, we present running times that actually split the dependency between these two quantities, and as a result may yield significant speedups.
In the online case, our results yield improved sample complexity bounds and allow for very efficient \textit{streaming} implementations with memory and processing-time requirements that are proportional to the size of a single sample.

On a high-level, our algorithms are based on a robust analysis of the classic idea of \emph{shift-and-invert} preconditioning \cite{saad1992numerical}, which allows us to efficiently reduce eigenvector computation to \emph{approximately} solving a \emph{short} sequence of \emph{well-conditioned} linear systems in $\lambda \bv I - \bv{A}^\top \bv {A}$ for some shift parameter $\lambda \approx \lambda_1(\bv{A})$. We then apply state-of-the-art stochastic gradient methods to approximately solve these linear systems.

\subsection{Our Approach}

The well known power method for computing the top eigenvector of $\bv{A}^\top \bv{A}$ starts with an initial vector $x$ and repeatedly multiplies by $\bv{A}^\top \bv{A}$, eventually causing $x$ to converge to the top eigenvector. For a random start vector, the power method converges in $O (\log(d/\epsilon)/\gap )$ iterations, where $\gap = (\lambda_1-\lambda_2)/\lambda_1$, $\lambda_i$ denotes the $i^{th}$ largest eigenvalue of $\bv{A}^\top \bv{A}$, and we assume a high-accuracy regime where $\epsilon < \gap$.
The dependence on this gap ensures that the largest eigenvalue is significantly amplified in comparison to the remaining values.

If the eigenvalue gap is small, one approach is replace $\bv{A}^\top \bv{A}$ with a preconditioned matrix -- i.e. a matrix with the same top eigenvector but a much larger gap.
Specifically, let $\bv{B} = \lambda \bv{ I} - \bv{A}^\top \bv{A}$ for some shift parameter $\lambda$. If $\lambda > \lambda_1$, we can see that the smallest eigenvector of $\bv{B}$ (the largest eigenvector of $\bv{B}^{-1}$) is equal to the largest eigenvector of $\bv{A}^\top \bv{A}$. Additionally, if $\lambda$ is close to $\lambda_1$, there will be a constant gap between the largest and second largest values of $\bv{B}^{-1}$. For example, if $\lambda = (1+\gap)\lambda_1$, then we will have $\lambda_1\left(\bv{B}^{-1} \right ) = \frac{1}{\lambda-\lambda_1} = \frac{1}{\gap \cdot \lambda_1}$ and $\lambda_2\left(\bv{B}^{-1} \right ) = \frac{1}{\lambda - \lambda_2} = \frac{1}{2\cdot\gap \cdot \lambda_1}$.

This constant factor gap ensures that the power method applied to $\bv{B}^{-1}$ converges to the top eigenvector of $\bv{A}^\top \bv{A}$ in just $O  (\log (d/\epsilon) )$ iterations. Of course, there is a catch -- each iteration of this \emph{shifted-and-inverted power method} must solve a linear system in $\bv{B}$, whose condition number is proportional $\frac{1}{\gap}$. For small gap, solving this system via iterative methods is more expensive.

Fortunately, linear system solvers are incredibly well studied and there are many efficient iterative algorithms we can adapt to apply $\bv{B}^{-1}$ approximately. In particular, we show how to accelerate the iterations of the shifted-and-inverted power method using variants of Stochastic Variance Reduced Gradient (SVRG) \cite{johnson2013accelerating}.
Due to the condition number of $\bv{B}$, we will not entirely avoid a $\frac{1}{\gap}$ dependence, however, we can separate this dependence from the input size $\nnz(\bv A)$.

Typically, stochastic gradient methods  are used to optimize convex functions that are given as the sum of many convex components. To solve a linear system $(\bv{M}^\top\bv{M}) x = b$ we minimize the convex function $f(x) = \frac{1}{2} x^\top (\bv{M}^\top\bv{M}) x - b^\top x$ with components $\psi_i(x) = \frac{1}{2} x^\top \left (m_i m_i^\top \right )x - \frac{1}{n}b^\top x$ where $m_i$ is the $i^{th}$ row of $\bv M$. Such an approach can be used to solve systems in $\bv{A}^\top \bv{A}$, however solving systems in $\bv{B} = \lambda \bv I - \bv A^\top \bv A$ requires more care. We require an analysis of SVRG that guarantees convergence even when some of our components are  \emph{non-convex}. We give a simple analysis for this setting, generalizing recent work in the area \cite{shalev2015sdca,csiba2015primal}.

Given fast approximate solvers for $\bv{B}$, the second main piece of our algorithmic framework is a new error bound for the shifted-and-inverted power method, showing that it is robust to approximate linear system solvers, such as SVRG. 
We give a general analysis, showing   
exactly what accuracy each system must be solved to, allowing for faster implementations using linear solvers with weaker guarantees. Our proofs center around the potential function: $G(x) \defeq
  \norm{\bv{P}_{v_1^{\perp}}x}_{\mb}/\norm{\bv{P}_{v_1}x}_{\mb}$,
where $\bv{P}_{v_1}$ and $\bv{P}_{v_1^{\perp}}$ are the projections onto the top eigenvector and its complement respectively. This function resembles tangent based potential functions used in previous work \cite{hardt2014noisy} except that we use the $\bv{B}$ norm rather than the $\ell_2$ norm. For the exact power method, this is irrelevant -- progress is identical in both norms (see Lemma \ref{same_progress} of the Appendix). However, $\norm{\cdot}_\mb$ is a natural norm for measuring the progress of linear system solvers for $\mb$, so our potential function makes it possible to extend analysis to the case when $\bv{B}^{-1}x$ is computed up to error $\xi$ with bounded $\norm{\xi}_\mb$.

\subsection{Our Results}

Our algorithmic framework described above offers several advantageous. We obtain improved running times for computing the top eigenvector in the offline model. In Theorem \ref{main_offline_theorem} we give an algorithm running in time $O\left(\left [\nnz(\bv{A}) + \frac{d \nrank \bv{A}}{\gap^2}\right ] \cdot \left [ \log \frac{1}{\epsilon} + \log^2 \frac{d}{\gap} \right ]\right)$, where $\nrank(\bv{A}) = \norm{\bv A}_F^2 / \norm{\bv A}_2^2 \le \rank(\bv{A})$ is the stable rank and $\nnz(\bv{A})$ is the number of non-zero entries. Up to log factors, our runtime is in many settings proportional to the input size $\nnz(\bv{A})$, and so is very efficient for large matrices. In the case when $\nnz(\bv{A}) \le \frac{d\nrank(\bv{A})}{\gap^2}$ we apply the results of \cite{frostig2015regularizing, lin2015catalyst} to provide an accelerated runtime of $O\left(\left [\frac{\nnz(\bv{A})^{3/4}(d \nrank (\bv{A}))^{1/4}}{\sqrt{\gap}}\right ] \cdot \left [\log \frac{d}{\gap} \log \frac{1}{\epsilon} + \log^3 \frac{d}{\gap} \right ]\right)$, shown in Theorem \ref{accelerated_offline_theorem}. Finally, in the case when $\epsilon > \gap$, our results easily extend to give gap-free bounds (Theorems \ref{main_gapfree_theorem} and \ref{acell_gapfree_theorem}), identical to those shown above but with $\gap$ replaced by $\epsilon$. Note that our offline results hold for any $\bv{A}$ and require no initial knowledge of the top eigenvector. In Section \ref{parameter_free} we discuss how to estimate the parameters $\lambda_1$, $\gap$, with modest additional runtime cost.

Our algorithms return an approximate top eigenvector $x$ with $x^\top\bv{A}^\top \bv{A} x \ge (1-\epsilon)\lambda_1$. By choosing error $\epsilon \cdot \gap$, we can ensure that $x$ is actually close to $v_1$ -- i.e. that $|x^\top v_1| \ge 1-\epsilon$. Further, we obtain the same asymptotic runtime since $O\left (\log \frac{1}{\epsilon \cdot \gap} + \log^2 \frac{d}{\gap} \right ) = O\left (\log \frac{1}{\epsilon} + \log^2 \frac{d}{\gap} \right )$. We compare our runtimes with previous work in Table \ref{offline_table}. 

In the online case, in Theorem \ref{warmstart_online_theorem}, we show how to improve an $O(\gap)$ approximation to the top eigenvector to an $\epsilon$ approximation with constant probability using $O\left (\frac{\nvar(\dist)}{\gap \cdot \epsilon}\right)$ samples where $\nvar(\dist)$ is a natural variance measure. Our algorithm is based on the streaming SVRG algorithm of \cite{frostig2014competing}. It requires just $O(d)$ amortized time per sample, uses just $O(d)$ space, and is easily parallelized. We can apply our result in a variety of regimes, using existing algorithms to obtain the initial $O(\gap)$ approximation and our algorithm to improve. As shown in Table \ref{online_table}, this gives improved runtimes and sample complexities over existing work. Notably, we give better asymptotic sample complexity than known matrix concentration results for general distributions, and give the first streaming algorithm that is asymptotically optimal in the popular Gaussian spike model.

Overall, our robust shifted-and-inverted power method analysis gives new understanding of this classical technique. It gives a means of obtaining provably accurate results when each iteration is implemented using fast linear system solvers with weak accuracy guarantees.
In practice, this reduction between approximate linear system solving and eigenvector computation shows that optimized regression libraries can be leveraged for faster eigenvector computation in many cases. Furthermore, in theory we believe that the reduction suggests computational limits inherent in eigenvector computation as seen by the often easier-to-analyze problem of linear system solving. Indeed, in Section \ref{sec:lower}, we provide evidence that in certain regimes our statistical results are optimal.


\subsection{Previous Work}

\subsubsection*{Offline Eigenvector Computation}\label{previous_work_offline}

Due to its universal applicability, eigenvector computation in the offline case is extremely well studied. Classical methods, such as the QR algorithm, take roughly $O(nd^2)$ time to compute a full eigendecomposition. This can be accelerated to $O(nd^{\omega - 1})$, where $\omega < 2.373$ is the matrix multiplication constant \cite{williams2012multiplying,le2014powers}, however this is still prohibitively expensive for large matrices. Hence, faster iterative methods are often employed, especially when only the top eigenvector (or a few of the top eigenvectors) is desired.

As discussed, the popular power method requires $O \left ( \frac{\log(d/\epsilon)}{\gap} \right )$ iterations to converge to an $\epsilon$ approximate top eigenvector. Using Chebyshev iteration, or more commonly, the Lanczos method, this bound can be improved to $O \left ( \frac{\log(d/\epsilon)}{\sqrt{\gap}} \right )$ \cite{saad1992numerical}, giving total runtime of $O \left (\nnz(\bv A) \cdot \frac{\log(d/\epsilon)}{\sqrt{\gap}} \right )$. When $\epsilon > \gap$, the $\gap$ terms in these runtimes can be replaced by $\epsilon$. While we focus on the high-precision regime when $\epsilon < \gap$, we also give gap-free bounds in Section \ref{sec:gapfree}.

Unfortunately, if $\nnz(\bv{A})$ is very large and $\gap$ is small, the above runtimes can still be quite expensive, and there is a natural desire to separate the $\frac{1}{\sqrt{\gap}}$ dependence from the $\nnz(\bv A)$ term. One approach is to use random subspace embedding matrices \cite{ailon2009fast,clarkson2013low} or fast row sampling algorithms \cite{cohen2015uniform}, which can be applied in $ O(\nnz(\bv A))$ time and yield a matrix $\bv{\tilde A}$ which is a good spectral approximation to the original. The number of rows in $\bv{\tilde A}$ depends only on the stable rank of $\bv A$ and the error of the embedding -- hence it can be significantly smaller than $n$. Applying such a subspace embedding and then computing the top eigenvector of $\bv{\tilde A}^\top \bv{\tilde A}$ requires runtime $O \left (\nnz(\bv A) + \poly(\nrank(\bv{A}),\epsilon,\gap) \right )$, achieving the goal of reducing runtime dependence on the input size $\nnz(\bv A)$. Unfortunately, the dependence on $\epsilon$ is significantly suboptimal -- such an approach cannot be used to obtain a linearly convergent algorithm. Further, the technique does not extend to the online setting, unless we are willing to store a full subspace embedding of our sampled rows. 

Another approach, which we follow more closely, is to apply stochastic optimization techniques, which iteratively update an estimate to the top eigenvector, considering a random row of $\bv{A}$ with each update step. Such algorithms naturally extend to the online setting and have led to improved dependence on the input size for a variety of problems \cite{bottou2010large}. Using variance-reduced stochastic gradient techniques, \cite{shamir2015stochastic} achieves runtime $ O\left (\left (\nnz(\bv A) + \frac{dr^2n^2}{\gap^2\lambda_1^2} \right ) \cdot \log(1/\epsilon) \log\log(1/\epsilon) \right )$ for approximately computing the top eigenvector of a matrix with constant probability.  Here $r$ is an upper bound on the squared row norms of $\bv{A}$. In the \emph{best case}, when row norms are uniform, this runtime can be simplified to $O\left (\left (\nnz(\bv A) + \frac{d\nrank(\bv A)^2}{\gap^2} \right ) \cdot \log(1/\epsilon) \log\log(1/\epsilon) \right )$.

The result in \cite{shamir2015stochastic} makes an important contribution in separating input size and gap dependencies using stochastic optimization techniques.
Unfortunately, the algorithm requires an approximation to the eigenvalue gap and a starting vector that has a constant dot product with the top eigenvector. In \cite{shamir2015fast} the analysis is extended to a random initialization, however loses polynomial factors in $d$. Furthermore, the dependencies on the stable rank and $\epsilon$ are suboptimal -- we improve them to $\nrank(\bv A)$ and $\log(1/\epsilon)$ respectively, obtaining true linear convergence.

\begin{table}[h]
\def\arraystretch{1}
\begin{center}
\begin{tabular}{|>{\centering}m{6.5cm}|c|}
\hline
\textbf{Algorithm}  & \textbf{Runtime}\\
\hline
Power Method & $O\left (\nnz(\bv{A})\frac{\log(d/\epsilon)}{\gap} \right )$ \\
\hline
Lanczos Method & $O\left (\nnz(\bv{A})\frac{\log(d/\epsilon)}{\sqrt{\gap}} \right )$ \\
\hline
Fast Subspace Embeddings \cite{clarkson2013low} Plus Lanczos  & $O\left (\nnz(\bv{A}) + \frac{d\nrank(\bv{A})}{\max \left \{ \gap^{2.5} \epsilon, \epsilon^{2.5} \right \} } \right )$ \\
\hline
SVRG \cite{shamir2015stochastic} (assuming bounded row norms, warm-start)& $O\left (\left (\nnz(\bv A) + \frac{d\nrank(\bv A)^2}{\gap^2} \right ) \cdot \log(1/\epsilon) \log\log(1/\epsilon) \right )$\\
\hline
\textbf{Theorem \ref{main_offline_theorem}} & $O\left(\left [\nnz(\bv{A}) + \frac{d \nrank(\bv{A})}{\gap^2}\right ] \cdot \left [\log \frac{1}{\epsilon} + \log^2 \frac{d}{\gap} \right ]\right)$\\
\hline
\textbf{Theorem \ref{accelerated_offline_theorem}} & $O\left(\left [\frac{\nnz(\bv{A})^{3/4}(d \nrank(\bv{A}))^{1/4}}{\sqrt{\gap}}\right ] \cdot   \left[\log \frac{d}{\gap} \log \frac{1}{\epsilon} + \log^3 \frac{d}{\gap} \right ]\right)$\\
\hline
\end{tabular}
\caption{Comparision to previous work on Offline Eigenvector Estimation. We give runtimes for computing a unit vector $x$ such that $x^\top \bv{A}^\top\bv{A} x \ge (1-\epsilon)\lambda_1$ in the regime $\epsilon = O(\gap)$.}
\label{offline_table}
\end{center}
\vspace{-2.5em}
\end{table}

\subsubsection*{Online Eigenvector Computation}

While in the offline case the primary concern is computation time, in the online, or statistical setting, research also focuses on minimizing the number of samples that are drawn from $\dist$ in order to achieve a given accuracy. Especially sought after are results that achieve asymptotically optimal accuracy as the sample size grows large.

While the result we give in Theorem \ref{warmstart_online_theorem} works for any distribution parameterized by a variance bound, in this section, in order to more easily compare to previous work, we normalize $\lambda_1 = 1$ and assume we have the row norm bound $\norm{a}_2^2 \le O(d)$ which then gives us the variance bound $\norm{\E_{a\sim\dist}\left[(aa^\top)^2\right]}_2 = O(d)$. Additionally, we compare runtimes for computing some $x$ such that $|x^\top v_1| \ge 1-\epsilon$, as this is the most popular guarantee studied in the literature. Theorem \ref{warmstart_online_theorem} is easily extended to this setting as obtaining $x$ with $x^T \bv{AA}^\top x \ge (1-\epsilon\cdot \gap) \lambda_1$ ensures $|x^\top v_1| \ge 1-\epsilon$. Our algorithm requires $O\left ( \frac{d}{\gap^2 \epsilon} \right )$ samples to find such a vector under the assumptions given above.

 The simplest algorithm in this setting is to take $n$ samples from $\dist$ and compute the leading eigenvector of the empirical estimate $\widehat\E[a a^\top] = \frac{1}{n} \sum_{i=1}^n a_i a_i^\top$. By a matrix Bernstein bound, such as inequality of Theorem 6.6.1 of \cite{tropp2015introduction}, $O\left ( \frac{d\log d}{\gap^2 \epsilon} \right )$ samples is enough to insure $\norm{\widehat\E[a a^\top] - \E[a a^\top]}_2 \le \sqrt{\epsilon}\cdot \gap$. 
By Lemma \ref{spectral_error_conversion} in the Appendix, this gives that, if $x$ is set to the top eigenvector of $\widehat\E[a a^\top] $ it will satisfy $|x^\top v_1| \ge 1-\epsilon$. $x$ can be approximated with any offline eigenvector algorithm.

A large body of work focuses on improving this simple algorithm, under a variety of assumptions on $\dist$. A common focus is on obtaining \emph{streaming algorithms}, in which the storage space is just $O(d)$ - proportional to the size of a single sample.
In Table \ref{online_table} we give a sampling of results in this area. All listed results rely on distributional assumptions at least as strong as those given above.

Note that, in each setting, we can use the cited algorithm to first compute an $O(\gap)$ approximate eigenvector, and then refine this approximation to an $\epsilon$ approximation using $O\left ( \frac{d}{\gap^2 \epsilon} \right )$ samples by applying our streaming SVRG based algorithm. 
This allows us to obtain improved runtimes and sample complexities. 
To save space, we do not include our improved runtime bounds in Table \ref{online_table}, however they are easy to derive by adding the runtime required by the given algorithm to achieve $O(\gap)$ accuracy, to $O\left (\frac{d^2}{\gap^{2} \epsilon}\right)$ -- the runtime required by our streaming algorithm.

The bounds given for the simple matrix Bernstein based algorithm described above,  Krasulina/Oja's Algorithm \cite{balsubramani2013fast}, and SGD \cite{shamir2015convergence}  require no additional assumptions, aside from those given at the beginning of this section.
The streaming results cited for \cite{mitliagkas2013memory} and \cite{hardt2014noisy} assume $a$ is generated from a
Gaussian spike model, where $a_i = \sqrt{\lambda_1}\gamma_i{v_1} + Z_i$
and $\gamma_i \sim \mathcal{N}(0, 1), Z_i \sim \mathcal{N}(0, I_d)$. We note that under this model, the matrix Bernstein results improve by a $\log d$ factor and so match our results in achieving asymptotically optimal convergence rate. The results of \cite{mitliagkas2013memory} and \cite{hardt2014noisy} sacrifice this optimality in order to operate under the streaming model. Our work gives the best of both works -- a streaming algorithm giving asymptotically optimal results. 

The streaming Alecton algorithm \cite{sa2015global} assumes $\E\norm{aa^\top \bv W a a^\top} \le O(1)\text{tr}(\bv W)$ for any symmetric $\bv W$ that commutes with $\E aa^\top$. This is strictly stronger than our assumption that 
\\$\norm{\E_{a\sim\dist} \left[(aa^\top)^2\right]}_2 = O(d)$.

\begin{table}[h]
\def\arraystretch{1}
\small
\begin{center}
\begin{tabular}{|c|c|c|c|c|}
\hline
\textbf{Algorithm} & \specialcell{\textbf{Sample}\\ \textbf{Size}} & \textbf{Runtime} & \textbf{Streaming?} & \specialcell{\textbf{Our Sample}\\ \textbf{Complexity}}\\ 
\hline
\specialcell{Matrix Bernstein plus \\ Lanczos (explicitly forming\\ sampled matrix)} & $O\left (\frac{d\log d}{gap^2 \epsilon}\right) $ & $O\left(\frac{d^3\log d}{gap^2\epsilon}\right)$ & $\times$ & $O\left (\frac{d\log d}{gap^3} + \frac{d}{gap^2 \epsilon}\right) $ \\ 
\hline
\specialcell{Matrix Bernstein plus \\Lanczos (iteratively applying\\ sampled matrix)} & $O\left (\frac{d\log d}{gap^2 \epsilon}\right) $ & $O\left (\frac{d^2\log d\cdot \log(d/\epsilon)}{gap^{2.5}\epsilon}\right)$ & $\times$ & $O\left (\frac{d\log d}{gap^3} + \frac{d}{gap^2 \epsilon}\right) $\\ 
\hline
\specialcell{Memory-efficient PCA \\\cite{mitliagkas2013memory, hardt2014noisy}} & $O\left(\frac{d\log (d/\epsilon)}{gap^3 \epsilon}\right )$ & $O\left (\frac{d^2\log (d/\epsilon)}{gap^3 \epsilon}\right)$ &  $\surd$ & $O\left(\frac{d\log (d/\gap)}{gap^4} + \frac{d}{gap^2 \epsilon}\right )$\\ 
\hline
Alecton \cite{sa2015global} & $O(\frac{d\log (d/\epsilon)}{gap^2 \epsilon})$ & $O(\frac{d^2\log (d/\epsilon)}{gap^2 \epsilon})$ & $\surd$ & $O(\frac{d\log (d/\gap)}{gap^3} + \frac{d}{gap^2 \epsilon})$ \\ 
\hline
\specialcell{Krasulina / Oja's \\Algorithm \cite{balsubramani2013fast}} & $O(\frac{d^{c_1}}{gap^2 \epsilon^{c_2}})$ &$O(\frac{d^{c_1+1}}{gap^2 \epsilon^{c_2}})$  & $\surd$ & $O(\frac{d^{c_1}}{gap^{2+c_2}} + \frac{d}{gap^2 \epsilon})$ \\ 
\hline
SGD \cite{shamir2015convergence} & $O(\frac{d^3\log (d/\epsilon)}{\epsilon^2})$ & $O(\frac{d^4\log (d/\epsilon)}{\epsilon^2})$ & $\surd$ & $O\left (\frac{d^3\log (d/\gap)}{\gap^2} +\frac{d}{\gap^2 \epsilon} \right ) $ \\
\hline
\end{tabular}
\caption{Summary of existing work on Online Eigenvector Estimation and improvements given by our results. Runtimes are for computing a unit vector $x$ such that $|x^\top v_1| \ge 1-\epsilon$. For each of these results we can obtain improved running times and sample complexities by running the algorithm to first compute an $O(\gap)$ approximate eigenvector, and then running our algorithm to obtain an $\epsilon$ approximation using an additional $O\left ( \frac{d}{\gap^2 \epsilon} \right )$ samples, $O(d)$ space, and $O(d)$ work per sample.
	}
\label{online_table}
\end{center}
\end{table}
\vspace{-2.5em}
\subsection{Paper Organization}
\begin{description}
\item[Section \ref{prelims}] Review problem definitions and parameters for our runtime and sample bounds.
\item[Section \ref{framework}] Describe the shifted-and-inverted power method and show how it can be implemented using approximate system solvers.
\item[Section \ref{sec:offline}] Show how to apply SVRG to solve systems in our shifted matrix, giving our main runtime results for offline eigenvector computation.
\item[Section \ref{sec:online}] Show how to use an online variant of SVRG to run the shifted-and-inverted power method, giving our main sampling complexity and runtime results in the statistical setting.
\item[Section \ref{parameter_free}] Show how to efficiently estimate the shift parameters required by our algorithms.
\item[Section \ref{sec:lower}] Give a lower bound in the statistical setting, showing that our results are asymptotically optimal for a wide parameter range.
\item[Section \ref{sec:gapfree}] Give gap-free runtime bounds, which apply when $\epsilon > \gap$.
\end{description}
\section{Preliminaries}\label{prelims}

We bold all matrix variables. 
We use $[n] \defeq \{1,...,n\}$. For a symmetric positive semidefinite (PSD) matrix $\bv{M}$ we let $\norm x_{\bv{M}} \defeq \sqrt{x^{\top}\bv{M} x}$ and $\lambda_1(\bv M), ..., \lambda_d(\bv M)$ denote its eigenvalues in decreasing order.
We use $\bv{M} \preceq \bv{N}$ to denote the condition that $x^{\top}\bv{M} x\leq x^{\top}\bv{N} x$
for all $x$. 

\subsection{The Offline Problem }

We are given a matrix $\ma\in\R^{n\times d}$ with rows $a^{(1)},...,a^{(n)}$
and wish to compute an approximation to the top eigenvector of $\mSigma\defeq\ma^{\top}\ma$. Specifically, for error parameter $\epsilon$ we want a unit vector $x$ such that $x^\top \bv{\Sigma} x \ge (1-\epsilon) \lambda_1(\bv{\Sigma})$.

\subsection{The Statistical Problem}
We have access to an oracle returning independent samples from a distribution $\dist$
on $\R^{d}$ and wish to compute the top eigenvector of $\mSigma\defeq\E_{a\sim\dist} \left [aa^\top \right ]$. Again, for error parameter $\epsilon$ we want to return a unit vector $x$ such that $x^\top \bv{\Sigma} x \ge (1-\epsilon) \lambda_1(\bv{\Sigma})$.

\subsection{Problem Parameters}

We parameterize the running times and sample complexities of our algorithms in terms of several natural properties of $\ma$, $\dist$, and $\mSigma$. Let $\lambda_{1},...,\lambda_{d}$ denote the eigenvalues of $\mSigma$
in decreasing order and $v_1,... ,v_d$ denote their corresponding
eigenvectors. We define the \emph{eigenvalue gap} by $\gap\defeq\frac{\lambda_{1}-\lambda_{2}}{\lambda_{1}}$.

We use the following additional parameters for the offline and statistical problems respectively:

\begin{itemize}
	\item \textbf{Offline Problem}: Let  $\nrank(\ma) \defeq\sum_{i}\frac{\lambda_{i}}{\lambda_{1}} = \frac{\norm{\ma}_F^2}{\norm{\ma}_2^2}$ denote the stable rank of $\ma$. Note that we always have $\nrank(\ma) \le \rank(\ma)$. Let $\nnz(\ma)$ denote the number of non-zero entries in $\ma$.
	\item \textbf{Online Problem}: Let $\nvar(\dist) \defeq \frac{\norm{\E_{a \sim \dist}\left [ \left (a a^\top \right )^2 \right ]}_2}{\norm{\E_{a \sim \dist} (aa^\top)}_2^2} = \frac{\norm{\E_{a \sim \dist}\left [ \left (a a^\top \right )^2 \right ]}_2}{\lambda_1^2}$ denote a natural upper bound on the variance of $\dist$ in various settings. Note that $\nvar(\dist) \ge 1$.
\end{itemize}


\section{Algorithmic Framework}\label{framework}

Here we develop our robust shift-and-invert framework. In Section~\ref{sec:framework:basics} we provide a basic overview of the framework and in Section~\ref{sec:framework:potential} we introduce the potential function we use to measure progress of our algorithms. In Section~\ref{sec:framework:power-iteration} we show how to analyze the framework given access to an exact linear system solver and in Section~\ref{sec:framework:approximate-power} we strengthen this analysis to work with an inexact linear system solver. Finally, in Section~\ref{sec:framework:init} we discuss initializing the framework.

\subsection{Shifted-and-Inverted Power Method Basics}
\label{sec:framework:basics}

We let $\mb_{\lambda}\defeq\lambda\mI-\mSigma$ denote the shifted matrix that we will use in our implementation of the shifted-and-inverted power method. As discussed, in order for $\bv{B}_\lambda^{-1}$ to have a large eigenvalue gap, $\lambda$ should be set to $(1+c \cdot \gap) \lambda_1$ for some constant $c \geq 0$. Throughout this section we assume that we have a crude estimate of $\lambda_1$ and $\gap$ and fix $\lambda$ to be a value satisfying  $\left (1 + \frac{\gap}{150}\right)\lambda_1 \le \lambda \le \left (1 + \frac{\gap}{100}\right)\lambda_1$. (See Section~\ref{parameter_free} for how we can compute such a $\lambda$). For the remainder of this section we work with such a fixed value of $\lambda$ and therefore for convenience denote $\bv{B}_\lambda$ as $\bv{B}$.

Note that $\lamiBinv{i} = \frac{1}{\lambda_i(\mb)} = \frac{1}{\lambda-\lambda_i}$ and so $\frac{\lamiBinv{1}}{\lamiBinv{2}} = \frac{\lambda-\lambda_2}{\lambda-\lambda_1} \ge \frac{\gap}{\gap/100} = 100.$ This large gap will ensure that, assuming the ability to apply $\bv{B}^{-1}$, the power method will converge very quickly. In the remainder of this section we develop our error analysis for the shifted-and-inverted power method which demonstrates that approximate application of $\bv{B}^{-1}$ in each iteration in fact suffices.

\subsection{Potential Function}
\label{sec:framework:potential}

Our analysis of the power method focuses on the objective of maximizing
the Rayleigh quotient, $x^\top \mSigma x$ for a unit vector $x$. Note that as the following lemma shows, this has a direct correspondence to the error in maximizing $|v_1^\top x|$:

\begin{lemma}[Bounding Eigenvector Error by Rayleigh Quotient]
\label{lem:ray_to_evec}
For a unit vector $x$ let $\epsilon = \lambda_1 - x^\top \bv{\Sigma} x$. If $\epsilon \leq \lambda_1 \cdot \gap$ then
\[
\left | v_1^\top x \right |
\geq \sqrt{1 - \frac{\epsilon}{\lambda_1 \cdot \gap}}.
\]
\end{lemma}

\begin{proof}
Among all unit vectors $x$ such that $\epsilon = \lambda_1 - x^\top \mSigma x$, a minimizer of $\left | v_1^\top x\right |$ has the form $x = (\sqrt{1 - \delta^2}) v_1 + \delta v_2$ for some $\delta$. We have
\begin{align*}
\epsilon
&=
\lambda_1 - x^\top \bv \Sigma x
= 
\lambda_1 - \lambda_1 (1 - \delta^2) - \lambda_2 \delta^2
=
(\lambda_1 - \lambda_2) \delta^2.
\end{align*}
Therefore by direct computation,
\begin{align*}
\left | v_1^\top x \right | = \sqrt{1 - \delta^2} = \sqrt{1 - \frac{\epsilon}{\lambda_1-\lambda_2}} = \sqrt{1 - \frac{\epsilon}{\lambda_1 \cdot \gap}} ~ .
\end{align*}
\end{proof}

In order to track the progress of our algorithm we use a more
complex potential function than just the Rayleigh quotient error, $\lambda_1 - x^\top \bv{\Sigma}
x$.  Our potential function $G$ is defined for $x \neq 0$ by
\begin{align*}
	G(x) \defeq
  \frac{\norm{\bv{P}_{v_1^{\perp}}x}_{\mb}}
  {\norm{\bv{P}_{v_1}x}_{\mb}} 
\end{align*}
where $\bv{P}_{v_1}$ and $\bv{P}_{v_1^{\perp}}$ are the
projections onto $v_1$ and the subspace
orthogonal to $v_1$ respectively. Equivalently, we have that:
\begin{align}\label{gequiv}
	G(x) =
  \frac{\sqrt{\norm{x}_{\mb}^{2}-\left(v_1^{\top}\mb^{1/2}x\right)^{2}}}{\abs{v_1^\top
  \mb^{1/2}x}} =  \frac{\sqrt{\sum_{i\geq 2} \frac{\alpha_i^2}{ \lambda_{i}(\bv{B}^{-1})}}}{\sqrt{\frac{\alpha_1^2}{ \lambda_{1}(\bv{B}^{-1})}}}.
\end{align}
where $\alpha_i = v_i^\top x$.

When the Rayleigh quotient error $\epsilon = \lambda_1 - x^\top\mSigma
x$ of $x$ is small, we can show a strong relation between $\epsilon$  and $G(x)$. We prove this in two parts. We first give a technical lemma, Lemma~\ref{lem:potfunc1}, that we will use several times for bounding the numerator of $G$. We then prove the connection in Lemma~\ref{lem:rayquot-potential}. 

\begin{lemma} \label{lem:potfunc1} For a unit vector $x$ and $\epsilon = \lambda_1 - x^\top \mSigma x$ if $\epsilon \leq \lambda_1 \cdot \gap$ then 
\[
\epsilon \leq 
x^\top \mb x - (v_1^\top \mb x) (v_1^\top x)
\leq \epsilon \left (1 + 
 \frac{\lambda - \lambda_1}{\lambda_1 \cdot \gap}\right).
\]
\end{lemma}

\begin{proof}
Since $\mb = \lambda \mI - \mSigma$ and since $v_1$ is an eigenvector of 
$\mSigma$ with eigenvalue $\lambda_1$ we have
\begin{align*}
x^\top \mb x - (v_1^\top \mb x) (v_1^\top x)
&=
\lambda \norm{x}_2^2 - x^\top \mSigma x
- (\lambda v_1^\top x - v_1^\top \mSigma x) (v_1^\top x)
\\
&=
\lambda - \lambda_1 + \epsilon -
(\lambda v_1^\top x - \lambda_1 v_1^\top x) (v_1^\top x)
\\
&=
(\lambda - \lambda_1) \left(1 - (v_1^\top x)^2\right) + \epsilon.
\end{align*}
Now by Lemma~\ref{lem:ray_to_evec} we know that $| v_1^\top x |
\geq \sqrt{1 - \frac{\epsilon }{\lambda_1 \cdot \gap}}$, giving us the upper bound. Furthermore, since trivially $\left |v_1^\top x\right | \leq 1$ and $\lambda - \lambda_1 > 0$, we have the lower bound.
\end{proof}

\begin{lemma}[Potential Function to Rayleigh Quotient Error Conversion]\label{lem:rayquot-potential}
	For a unit vector $x$ and $\epsilon = \lambda_1 - x^\top \mSigma x$ if $\epsilon \leq \frac{1}{2}\lambda_1 \cdot \gap$, we have:
	\begin{align*}
		\frac{\epsilon}{\lambda-\lambda_1} \leq G(x)^2 \leq \left(1+\frac{\lambda-\lambda_1}{\lambda_1 \cdot \gap}\right) \left(1+\frac{2\epsilon}{\lambda_1 \cdot \gap}\right)\frac{\epsilon}{\lambda-\lambda_1}.
	\end{align*}
\end{lemma}
\begin{proof}
	Since $v_1$ is an eigenvector of $\mb$, we can write $G(x)^2 = \frac{x^\top \mb x - (v_1^\top \mb x) (v_1^\top x)}{(v_1^\top \mb x) (v_1^\top x)}$. Lemmas~\ref{lem:ray_to_evec} and~\ref{lem:potfunc1} then give us:
	\begin{align*}
		\frac{\epsilon}{\lambda-\lambda_1} \leq G(x)^2 \leq \left(1+\frac{\lambda-\lambda_1}{\lambda_1 \cdot \gap}\right)\frac{\epsilon}{\left(\lambda-\lambda_1\right) \left(1-\frac{\epsilon}{\lambda_1 \cdot \gap}\right)}.
	\end{align*}
	Since $\epsilon \leq \frac{1}{2}\lambda_1 \cdot \gap$, we have $\frac{1}{1-\frac{\epsilon}{\lambda_1 \cdot \gap}} \leq 1+\frac{2\epsilon}{\lambda_1 \cdot \gap}$. This proves the lemma.
\end{proof}

\subsection{Power Iteration}
\label{sec:framework:power-iteration}

Here we show that the shifted-and-inverted power iteration in fact makes progress with respect to our objective function given an exact linear system solver for $\mb$. Formally, we show that applying $\mb^{-1}$ to a vector $x$ decreases the potential function $G(x)$ geometrically.

\begin{theorem}\label{thm:powermethod}
	Let $x$ be a unit vector with $\inprod{x,v_1} \neq 0$ and let $\xtilde = \mb^{-1}x$, i.e. the power method update of $\mb^{-1}$ on $x$. Then, under our assumption on $\lambda$, we have:
	\begin{align*}
		G(\xtilde) \leq \frac{\lamiBinv{2}}{\lamiBinv{1}} G(x)
		\leq \frac{1}{100} G(x).
	\end{align*}
\end{theorem}
Note that $\xtilde$ may no longer be a unit vector. However, $G(\xtilde, v_1) = G(c \xtilde, v_1)$ for any scaling parameter $c$, so the theorem also holds for $\xtilde$ scaled to have unit norm.
\begin{proof}
	Writing $x$ in the eigenbasis, we have $x=\sum_i \alpha_i v_i$ and  $\xtilde = \sum_i \alpha_i \lamiBinv{i}v_i$. Since $\inprod{x,v_1} \neq 0$, $\alpha_1 \neq 0$  and by the equivalent formulation of $G(x)$ given in \eqref{gequiv}:
	\begin{align*}
		G(\xtilde) = \frac{\sqrt{\sum_{i\geq 2} \alpha_i^2 \lambda_{i}(\bv{B}^{-1})}}{\sqrt{\alpha_1^2 \lambda_{1}(\bv{B}^{-1})}}
		\leq \frac{\lamiBinv{2}}{\lamiBinv{1}} \cdot \frac{\sqrt{\sum_{i\geq 2} \frac{\alpha_i^2}{\lambda_{i}(\bv{B}^{-1})}}}{\sqrt{\frac{\alpha_1^2}{\lambda_1(\bv{B}^{-1})}}}
		= \frac{\lamiBinv{2}}{\lamiBinv{1}} \cdot G(x) ~.
	\end{align*}
Recalling that $\frac{\lamiBinv{1}}{\lamiBinv{2}} = \frac{\lambda-\lambda_2}{\lambda-\lambda_1} \ge \frac{\gap}{\gap/100} = 100$ yields the result. 
\end{proof}
The challenge in using the above theorem, and any traditional analysis of the shifted-and-inverted power method, is that we don't actually have access to $\mb^{-1}$. In the next section we show that the shifted-and-inverted power method is robust -- we still make progress on our objective function even if we only approximate $\mb^{-1}x$ using a fast linear system solver.

\subsection{Approximate Power Iteration}
\label{sec:framework:approximate-power}

We are now ready to prove our main result. We show that each iteration of the shifted-and-inverted power method makes constant factor expected progress on our potential function assuming we:
\begin{enumerate}
\item Start with a sufficiently good $x$ and an approximation of $\lambda_1$ 
\item Can apply $\bv{B}^{-1}$ approximately using a system solver such that the function error (i.e. distance to $\bv{B}^{-1} x$ in the $\bv{B}$ norm) is sufficiently small \emph{in expectation}.
\item Can estimate Rayleigh quotients over $\bv{\Sigma}$ well enough to only accept updates that do not hurt progress on the objective function too much.
\end{enumerate}

This third assumption is necessary since the second assumption is quite weak. An expected progress bound on the linear system solver allows, for example, the solver to occasionally return a solution that is entirely orthogonal to $v_1$, causing us to make unbounded backwards progress on our potential function. The third assumption allows us to reject possibly harmful updates and ensure that we still make progress in expectation. In the offline setting, we can access $\bv{A}$ and are able to compute Rayleigh quotients exactly in time $\nnz(\bv{A})$ time. However, we only assume the ability to estimate quotients since in the online setting we only have access to $\mSigma$ through samples from $\dist$.

Our general theorem for the approximate power iteration, Theorem~\ref{thm:powermethod-perturb}, assumes that we can solve linear systems to some absolute accuracy in expectation. This is not completely standard. Typically, system solver analysis assumes an initial approximation to $\bv{B}^{-1} x$ and then shows a relative progress bound -- that the quality
of the initial approximation is improved geometrically in each iteration of the algorithm. In Corollary~\ref{cor:constant_factor_corollary} we show how to find a  coarse initial approximation to $\bv{B}^{-1} x$, in fact just approximating $\bv{B}^{-1}$ with $\frac{1}{x^\top \mb x} x$. Using this approximation, we show that Theorem~\ref{thm:powermethod-perturb} actually implies that traditional system solver relative progress bounds suffice. 

Note that in both claims we measure error of the linear system solver using $\norm{\cdot}_\mb$. This is a natural norm in which geometric convergence is shown for many linear system solvers and directly corresponds to the function error of minimizing $f(w) = \frac{1}{2} w^\top \mb w - w^\top x$ to compute $\mb^{-1} x$.

\begin{theorem}[Approximate Shifted-and-Inverted Power Iteration -- Warm Start]\label{thm:powermethod-perturb}
	Let $x=\sum_i \alpha_i v_i$ be a unit vector such that $G(x) \leq \frac{1}{\sqrt{10}}$. Suppose we know some shift parameter $\lambda$ with $\left (1 + \frac{\gap}{150} \right ) \lambda_1< \lambda \le \left (1 + \frac{\gap}{100} \right ) \lambda_1$ and an estimate $\lambdah_1$ of $\lambda_1$ such that
	$
		\frac{10}{11} \left(\lambda-\lambda_1\right) \leq \lambda-\lambdah_1 \leq \lambda-\lambda_1
	$.
	Furthermore, suppose we have a subroutine $\mathrm{solve}(\cdot)$ such that on any input $x$
	\begin{align*}
		\expec{\norm{\solve{x} - \mb^{-1}x}_{\mb}} \leq  \frac{c_1}{1000}
		\sqrt{\lambda_1(\bv{B}^{-1})},
	\end{align*}
	for some $c_1 < 1$,
	and a subroutine $\rayquoth{\cdot}$ that on any input $x \neq 0$
	\begin{align*}
		\abs{\rayquoth{x} - \rayquot(x)}
			\leq \frac{1}{30}\left(\lambda-\lambda_1\right) \;  \text{ for all nonzero } \; x \in \R^d.
	\end{align*}
	where $\rayquot(x) \defeq \frac{x^\top \mSigma x}{x^\top x}$.

	Then the following update procedure:
	\begin{align*}
		\text {Set }\xhat = \solve{x},
	\end{align*}
	\begin{align*}
		\text {Set } \xtilde = \left\{
		\begin{array}{cc}
		\xhat& \mbox{ if } \left\{
		\begin{array}{c}
			\rayquoth{\xhat} \geq \lambdah_1 - \left(\lambda-\lambdah_1\right)/6 \mbox{ and } \\
			\norm{\xhat}_2 \geq \frac{2}{3}\frac{1}{\lambda-\lambdah_1}
		\end{array} \right. \\
		x & \mbox{ otherwise,}
		\end{array}
		\right.
	\end{align*}
	satisfies the following:
	\begin{itemize}
		\item	$G(\xtilde)\leq \frac{1}{\sqrt{10}}$ and
		\item	$\expec{G(\xtilde)}\leq \frac{3}{25} G(x) + \frac{c_1}{500}$.
	\end{itemize}
\end{theorem}
That is, not only do we decrease our potential function by a constant factor in expectation, but we are guaranteed that the potential function will never increase beyond $1/\sqrt{10}$.
\begin{proof}
	The first claim follows directly from our choice of $\xtilde$ from $x$ and $\xhat$. 
	If $\xtilde = x$, it holds trivially by our assumption that $G(x) \leq \frac{1}{\sqrt{10}}$.  Otherwise, $\xtilde = \xhat$ and we know that 
	\begin{align*}
	\lambda_1 - \rayquot\left(\xhat\right) 
	&\leq \lambdah_1 - \rayquot\left(\xhat\right) \leq \lambdah_1 - \rayquoth{\xhat} + \abs{\rayquoth{\xhat} - \rayquot\left(\xhat\right)} 
	\\
	&\leq \frac{\lambda-\lambdah_1}{6} + \frac{\lambda-\lambda_1}{30} 
	\leq \frac{\lambda-\lambda_1}{5} \le \frac{\lambda_1 \cdot \gap}{500} ~.
	\end{align*} 
	The claim then follows from Lemma~\ref{lem:rayquot-potential} as
	\begin{align*}
		G(\xhat)^2 
		&\leq 
		\left(1 + \frac{\lambda - \lambda_1}{\lambda_1 \cdot \gap}\right) 
		\left(1 + \frac{2\left(\lambda_1 - \rayquot\left(\xhat\right)\right)}{\lambda_1 \cdot \gap}\right)
		\frac{\lambda_1 - \rayquot\left(\xhat\right)}{\lambda-\lambda_1} \\
		&\leq \frac{101}{100}\cdot \frac{251}{250} \cdot
		\frac{\left(\frac{\lambda_1 \cdot \gap}{500}\right)}
		{\left(\frac{\lambda_1 \cdot \gap}{150}\right)} 
		\le \frac{1}{\sqrt{10}} ~.
	\end{align*}
	
	All that remains is to show the second claim, that $\expec{G(\xtilde)}\leq \frac{3}{25} G(x) + \frac{4c_1}{1000}$. Let $\mcalF$ denote the event that we accept our iteration and set $x= \xhat = \solve{x}$. That is:
	\begin{align*}
	\mcalF &\defeq \set{\rayquoth{\xhat} \geq \lambdah_1 - 
		\frac{\lambda - \lambdah_1}{6}} \cup \set{\norm{\xhat}_2 \geq \frac{2}{3}\frac{1}{\lambda-\lambdah_1}}.
	\end{align*}
	Using our bounds on $\lambdah_1$ and $\rayquoth{\cdot}$, we know that $\rayquoth{x} \leq \rayquot(x) + (\lambda - \lambda_1) / 30$ and $\lambda - \lambdah_1 \leq  \lambda - \lambda_1$.
	Therefore, since $-1/6 - 1/30 \geq -1/2$ we have
	\begin{align*}
	\mathcal{F} &\subseteq \set{\rayquot\left(\xhat\right) \geq \lambda_1 - \left(\lambda-\lambda_1\right)/2} \cup \set{\norm{\xhat}_2 \geq \frac{2}{3}\frac{1}{\lambda-\lambda_1}},
	\end{align*}
	
	We will complete the proof in two steps. First we let $\xi \defeq \xhat - \mb^{-1}x$ and show that assuming $\mcalF$ is true then $G(\xhat)$ and
    $\norm{\xi}_{\mb}$ are linearly related, i.e. expected error bounds on $\norm{\xi}_{\mb}$ correspond to expected error bounds on $G(\xhat)$. Second, we bound the probability that $\mcalF$ does not occur and bound error incurred in this case. Combining these yields the result. 
    
    To show the linear relationship in the case where $\mcalF$ is true, first note Lemma~\ref{lem:ray_to_evec} shows that in this case $\abs{v_1^\top \frac{\xhat}{\norm{\xhat}_2}} \geq \sqrt{1-\frac{\lambda_1-\rayquot(\xhat)}{\lambda_1 \cdot \gap}} \geq \frac{3}{4}$. Consequently,
    \[
    \norm{\bv{P}_{v_1}\xhat}_{\mb}
    = \abs{v_1^\top\xhat}\sqrt{\lambda-\lambda_1}
    = \abs{v_1^\top \frac{\xhat}{\norm{\xhat}_2}}\cdot \norm{\xhat}\sqrt{\lambda-\lambda_1} 
     \geq \frac{3}{4}\cdot \frac{2}{3} \frac{1}{\sqrt{\lambda-\lambda_1}}
     =
     \frac{\sqrt{\lambda_1(\bv{B}^{-1})}}{2} ~.
    \]
    However, 
    \[
    \norm{\bv{P}_{v_1^{\perp}}\xhat}_{\mb}
	\leq \norm{\bv{P}_{v_1^{\perp}}\mb^{-1} x}_{\mb} + \norm{\bv{P}_{v_1^{\perp}}\xi}_\mb
		\leq \norm{\bv{P}_{v_1^{\perp}}\mb^{-1} x}_{\mb} + \norm{\xi}_{\mb} 
	\]
	and by Theorem~\ref{thm:powermethod} and the definition of $G$ we have
	\[
	\norm{\bv{P}_{v_1^{\perp}}\mb^{-1} x}_{\mb}
	= \norm{\bv{P}_{v_1}\mb^{-1} x}_{\mb} \cdot G(\mb^{-1} x)
	\leq \left(\abs{\inprod{x,v_1}} \sqrt{\lambda_1(\bv{B}^{-1})} \right)
	\cdot \frac{G(x)}{100} ~.
	\]
	Taking expectations, using that $\abs{\inprod{x,v_1}} \leq 1$, and combining these three inequalities yields
	\begin{equation}
	\expec{G(\xhat) \middle\vert  \mcalF}
	=
	\expec{\frac{\norm{\bv{P}_{v_1^{\perp}}\mb^{-1} x}_{\mb}}{\norm{\bv{P}_{v_1}\mb^{-1} x}_{\mb}} \middle\vert  \mcalF}
	\leq \frac{G(x)}{50}  
	+ 2 \frac{ \expec{\norm{\xi}_{\mb} 
			 \middle\vert \mcalF}}{\sqrt{\lambda_1(\bv{B}^{-1})}}
	\label{eqn:expec-potential}
	\end{equation}
	So, conditioning on making an update and changing $x$ (i.e. $\mathcal{F}$ occurring), we see that our potential function changes exactly as in the exact case (Theorem \ref{thm:powermethod}) with additional additive error due to our inexact linear system solve.

	Next we upper bound $\prob{\mcalF}$ and use it to compute $\expec{\norm{\xi}_{\mb}\middle\vert \mcalF}$. We will show that
	 $$\mcalG \defeq \set{\norm{\xi}_{\mb} \leq \frac{1}{100}\cdot \sqrt{\lamiBinv{1}} } \subseteq \mcalF$$ which then implies by Markov inequality that
	\begin{align}
		\prob{\mcalF} &\geq \prob{\norm{\xi}_{\mb} \leq \frac{1}{100}\cdot \sqrt{\lamiBinv{1}} } 
		\geq 1- \frac{\expec{\norm{\xi}_{\mb}}}{\frac{1}{100}\cdot \sqrt{\lamiBinv{1}} } \geq \frac{9}{10},\label{eqn:PF}
	\end{align}
	where we used the fact that $\E [\norm{\xi}_{\mb}] \le \frac{c_1}{1000} \sqrt{\lambda_1(\bv{B}^{-1})}$ for some $c_1 < 1$.
	
	Let us now show that $\mcalG \subseteq \mcalF$. Suppose $\mcalG$ is occurs. We can bound $\norm{\xhat}_2$ as follows:
	\begin{align}
		\norm{\xhat}_2 
		& \geq \norm{\mb^{-1}x}_2 - \norm{\xi}_2
		\geq \norm{\mb^{-1}x}-\sqrt{\lamiBinv{1}}\norm{\xi}_{\mb} \nonumber \\
		&\geq \abs{\alpha_1}\lamiBinv{1} - \frac{1}{100}\cdot \lamiBinv{1} \nonumber \\
		&= \frac{1}{\lambda-\lambda_1} \left(\abs{\alpha_1} - \frac{1}{100}\right)
		\geq \frac{3}{4}\frac{1}{\lambda-\lambda_1},\label{eqn:f2}
	\end{align}
	where we use Lemmas~\ref{lem:potfunc1} and~\ref{lem:rayquot-potential} to conclude that $\abs{\alpha_1} \geq \sqrt{1-\frac{1}{10}}$. We now turn to showing the Rayleigh quotient condition required by $\mcalF$. In order to do this, we first bound $\xhat^\top \mb \xhat - \left(v_1^\top \mb \xhat\right)\left(v_1^\top \xhat\right)$ and then use Lemma~\ref{lem:potfunc1}. We have:
	\begin{align*}
	\sqrt{\xhat^\top \mb \xhat - \left(v_1^\top \mb \xhat\right)\left(v_1^\top \xhat\right)} 
	&= \norm{\bv{P}_{v_1^{\perp}}\xhat}_\mb 
	\leq \norm{\bv{P}_{v_1^{\perp}}\mb^{-1}x}_\mb + \norm{\bv{P}_{v_1^{\perp}} \xi}_\mb \\
	&\leq \sqrt{\sum_{i\geq 2} \alpha_i^2 \lamiBinv{i}} + \frac{1}{100}\cdot \sqrt{\lamiBinv{1}} \\
	&\leq \sqrt{\lamiBinv{2}} + \frac{1}{100}\cdot \sqrt{\lamiBinv{1}}
	\leq \frac{1}{9}\sqrt{\lambda-\lambda_1},
	\end{align*}
	where we used the fact that $\lamiBinv{2}\leq \frac{1}{100}\lamiBinv{1}$ since $\lambda \le \lambda_1 + \frac{\gap}{100}$ in the last step. Now, using Lemma~\ref{lem:potfunc1} and the bound on $\norm{\xhat}_2$, we conclude that
	\begin{align}
		\lambdah_1 - \rayquoth{\xhat}
		&\leq \lambda_1 - \rayquot\left(\xhat\right) + \abs{\rayquot\left(\xhat\right)-\rayquoth{\xhat}} + \lambdah_1 - \lambda_1\nonumber \\
		&\leq \frac{\xhat^\top \mb \xhat - \left(v_1^\top \mb \xhat\right)\left(v_1^\top \xhat\right)}{\norm{\xhat}^2_2} + \frac{\lambda-\lambda_1}{30} + \frac{\lambda-\lambda_1}{11} \nonumber \\
		&\leq \frac{1}{81\left(\lambda-\lambda_1\right)} \cdot \frac{16}{9}\left(\lambda-\lambda_1\right)^2 + \frac{\lambda-\lambda_1}{8} \nonumber \\
		&\leq \left(\lambda-\lambda_1\right)/6
		\leq \left(\lambda-\lambdah_1\right)/4. \label{eqn:f1}
	\end{align}
	Combining~\eqref{eqn:f2} and~\eqref{eqn:f1} shows that $\mcalG \subseteq \mcalF$ there by proving~\eqref{eqn:PF}.
	
	Using this and the fact that $\norm{\cdot}_{\mb} \geq 0$ we can upper bound $\expec{\norm{\xi}_{\mb}\middle\vert \mcalF}$ as follows:
	\begin{align*}
		\expec{\norm{\xi}_{\mb}\middle\vert \mcalF} \leq \frac{1}{\prob{\mcalF}} \cdot \expec{\norm{\xi}_{\mb}}
		\leq \frac{c_1}{900}\cdot \sqrt{\lambda_1(\bv{B}^{-1})}
	\end{align*}
	Plugging this into~\eqref{eqn:expec-potential}, we obtain:
	\[
	\expec{G(\xhat)\middle\vert \mcalF} 
	\leq 
	\frac{1}{50} G(x) + \frac{2\expec{\norm{\xi}_{\mb}\middle\vert \mcalF}}{\sqrt{\lambda_1(\bv{B}^{-1})}}
	\leq \frac{1}{50}\cdot G(x) + \frac{2c_1}{900}.
	\]
	We can now finally bound $\expec{G(\xtilde)}$ as follows:
	\begin{align*}
		\expec{G(\xtilde)} &= \prob{\mcalF} \cdot \expec{G(\xhat)\middle\vert \mcalF} + \left(1-\prob{\mcalF}\right) G(x) \\
		&\leq \frac{9}{10} \left ( \frac{1}{50}\cdot G(x) + \frac{2c_1}{900}\right) + \frac{1}{10} G(x) 
		= \frac{3}{25} G(x) + \frac{2c_1}{1000}.
	\end{align*}
	This proves the theorem.
\end{proof}

\begin{corollary} [Relative Error Linear System Solvers]
	\label{cor:constant_factor_corollary} 
		For any unit vector $x$, we have:
		\begin{equation}
		\label{eq:constant_factor_corollary}
		\norm{\frac{1}{x^\top \mb x} x - \mb^{-1}x}_{\mb} \leq \alpha_1\sqrt{\lambda_1(\bv{B}^{-1})} \cdot G(x) =  \lamiBinv{1} \sqrt{\sum_{i\geq 2} \frac{\alpha_i^2}{\lamiBinv{i}}},
		\end{equation}
		where $x=\sum_i \alpha_i v_i$ is the decomposition of $x$ along $v_i$.
	Therefore, instantiating Theorem \ref{thm:powermethod-perturb} with $c_1 = \alpha_1 G(x)$ gives $\E [G(\xtilde)] \le \frac{4}{25} G(x)$ as long as:
	\begin{align*}
		\expec{\norm{\solve{x}-\mb^{-1}x}_{\mb}} \le  \frac{1}{1000} \norm{\frac{1}{\lambda-x^\top \bv{\Sigma}x} x - \mb^{-1}x}_{\mb}.
	\end{align*}
\end{corollary}
\begin{proof}
	Since $\mb$ is PSD we see that if we let $f(w) = \frac{1}{2} w^\top \mb w - w^\top x$, then the minimizer is $\mb^{-1} x$. Furthermore note that $\frac{1}{x^\top \mb x} = \argmin_\beta
	f(\beta x)$ and therefore
	\begin{align*}
	 \norm{\frac{1}{x^\top \mb x} x - \mb^{-1}x}_{\mb}^2
	&=  x^\top \mb^{-1} x - \frac{1}{x^\top \mb x}
	= 2 \left [f\left(\frac{x}{x^\top \mb x}\right) - f(\mb^{-1} x) \right] \\
	= & 2 \left [\min_\beta f(\beta x) - f(\mb^{-1} x) \right]
	\le 2 \left[f(\lamiBinv{1}x) - f(\mb^{-1} x) \right] \\
	= &\lamiBinv{1}^2 x^\top \mb x - 2 \lamiBinv{1} x^\top x + x^\top \mb^{-1} x \\
	= & \sum_{i=1}^d \abs{v_i^\top \mb^{\frac{1}{2}} x}^2(\lamiBinv{1}-\lamiBinv{i})^2
	\le \lamiBinv{1}^2\sum_{i\ge 2} \abs{v_i^\top \mb^{\frac{1}{2}} x}^2 \\
	= & \lamiBinv{1}^2\sum_{i\geq 2} \frac{\alpha_i^2}{\lamiBinv{i}},
	\end{align*}
	which proves \eqref{eq:constant_factor_corollary}.	
Consequently
\begin{align*}
\frac{c_1}{1000}\sqrt{\lambda_1(\bv{B}^{-1})} = \frac{1}{1000} \alpha_1 G(x)\sqrt{\lambda_1(\bv{B}^{-1})} \ge \frac{1}{1000} \norm{\frac{1}{x^\top \mb  x} x - \mb^{-1}x}_{\mb}
\end{align*}
which with Theorem~\ref{thm:powermethod-perturb} then completes the proof.
\end{proof}

\subsection{Initialization}
\label{sec:framework:init}

Theorem \ref{thm:powermethod-perturb} and Corollary \ref{cor:constant_factor_corollary} show that, given a good enough approximation to $v_1$, we can rapidly refine this approximation by applying the shifted-and-inverted power method. In this section, we cover initialization. That is, how to obtain a good enough approximation to apply these results.


We first give a simple bound on the quality of a randomly chosen start vector $x_0$. 

\begin{lemma}[Random Initialization Quality]\label{lem:init-random}
    Suppose $x \sim \mathcal{N}(0, \bv{I})$, and we initialize $x_0$ as $\frac{x}{\norm{x}_2}$, then with probability greater than $1- O\left (\frac{1}{d^{10}} \right )$, we have:
    \begin{equation*}
    	G(x_0) \le \sqrt{\kappa(\mb^{-1})} d^{10.5} \le 15 \frac{1}{\sqrt{\gap}} \cdot d^{10.5}
    \end{equation*}
    where $\kappa(\mb^{-1}) = \lambda_1(\mb^{-1}) / \lambda_d(\mb^{-1)}$.
\end{lemma}

\begin{proof}
	\begin{align*}
		G(x_0) =& G(x) = \frac{\norm{\bv{P}_{v_1^{\perp}}x}_{\mb}} {\norm{\bv{P}_{v_1}x}_{\mb}} = \frac{\sqrt{\norm{x}_{\mb}^{2}-\left(v_1^\top\mb^{1/2}x\right)^{2}}}{\abs{v_1^\top \mb^{1/2}x}}
		=\frac{\sqrt{\sum_{i\geq2} \frac{(v_i^\top x)^2}{\lamiBinv{i}}}}{\sqrt{\frac{(v_1^\top x)^2}{\lamiBinv{1}}}}, \\
		\le& \sqrt{\kappa(\mb^{-1})} \cdot \frac{\sqrt{\sum_{i\geq2}(v_i^\top x)^2}}{\abs{v_1^\top x}}
	\end{align*}
	Since $\{v_i^\top x\}_i$ are independent standard normal Gaussian variables.
	By standard concentration arguments, with probability greater than $1- e^{-\Omega(d)}$, we have $\sqrt{\sum_{i\geq2}(v_i^\top x)^2} = O(\sqrt{d})$. Meanwhile, $v_1^\top x$ is just a one-dimensional standard Gaussian. It is easy to show $\Pr \left(\abs{v_1^\top x} \le \frac{1}{d^{10}} \right) = O\left (\frac{1}{d^{10}} \right)$, which finishes the proof.
\end{proof}

We now show that we can rapidly decrease our initial error to obtain the required $G(x) \le \frac{1}{\sqrt{10}}$ bound for Theorem \ref{thm:powermethod-perturb}.

\begin{theorem}[Approximate Shifted-and-Inverted Power Method -- Burn-In]\label{thm:init-offline}
	Suppose we initialize $x_0$ as in Lemma \ref{lem:init-random} and suppose we have access to a subroutine $\solve{\cdot}$ such that
	\begin{align*}
		\expec{\norm{\solve{x}-\mb^{-1}x}_{\mb}} \leq \frac{1}{3000 \kappa(\mb^{-1})d^{21}} \cdot \norm{\frac{1}{\lambda-x^\top \bv \Sigma x} x - \mb^{-1}x}_{\mb}
	\end{align*}
	where $\kappa(\mb^{-1}) = \lambda_1(\mb^{-1}) / \lambda_d(\mb^{-1)}$.
	Then the following procedure,
	\begin{align*}
		x_{t} = \solve{x_{t-1}}/\norm{\solve{x_{t-1}}}_2
	\end{align*}
	after $T = O\left(\log d + \log \kappa(\mb^{-1}))\right)$ iterations  satisfies:
	\begin{align*}
		G(x_T) \leq \frac{1}{\sqrt{10}},
	\end{align*}
	with probability greater than $1- O(\frac{1}{d^{10}})$.
\end{theorem}

\begin{proof}
As before, we first bound the numerator and denominator of $G(\xhat)$ more carefully as follows:
\begin{align*}
	\begin{array}{lrl}
	\textbf{Numerator:}	&\norm{\bv{P}_{v_1^{\perp}}\xhat}_{\mb}
	&\leq \norm{\bv{P}_{v_1^{\perp}}\mb^{-1} x}_{\mb} + \norm{\bv{P}_{v_1^{\perp}}\xi}_\mb
	\leq \norm{\bv{P}_{v_1^{\perp}}\mb^{-1} x}_{\mb} + \norm{\xi}_{\mb} \\
	& &=\sqrt{\sum_{i\geq2}\left(v_{i}^{T}B^{-1/2}x\right)^{2}}+\norm{\xi}_{\mb} = \sqrt{\sum_{i\geq2} \alpha_{i}^{2}\lamiBinv{i}} + \norm{\xi}_{\mb}, \\
	\textbf{Denominator:} &\norm{\bv{P}_{v_1}\xhat}_{\mb}
	&\geq  \norm{\bv{P}_{v_1}\mb^{-1} x}_{\mb} - \norm{\bv{P}_{v_1}\xi}_\mb
	\geq \norm{\bv{P}_{v_1}\mb^{-1} x }_{\mb} - \norm{\xi}_{\mb} \\
	&&=\abs{v_{i}^{T}{\mb}^{-1/2}x}-\norm{\xi}_{\mb} = \alpha_1\sqrt{\lamiBinv{1}} - \norm{\xi}_{\mb}
	\end{array}
\end{align*}

We now use the above estimates to bound $G(\xhat)$.
\begin{align*}
	G(\xhat)  &\leq \frac{\sqrt{\sum_{i\geq2} \alpha_{i}^{2}\lamiBinv{i}} + \norm{\xi}_{\mb}}{\alpha_1\sqrt{\lamiBinv{1}} - \norm{\xi}_{\mb}} 
	\leq \frac{\lamiBinv{2}\sqrt{\sum_{i\geq2} \frac{\alpha_{i}^{2}}{\lamiBinv{i}}} + \norm{\xi}_{\mb}}{\lamiBinv{1}\sqrt{\frac{\alpha_1^2}{\lamiBinv{1}}} - \norm{\xi}_{\mb}} \\
	&= G(x) \frac{\lamiBinv{2} + \norm{\xi}_{\mb}/\sqrt{\sum_{i\geq2} \frac{\alpha_{i}^{2}}{\lamiBinv{i}}}}{\lamiBinv{1} - \norm{\xi}_{\mb}/\sqrt{\frac{\alpha_1^2}{\lamiBinv{1}}}}
\end{align*}
By Lemma \ref{lem:init-random}, we know with at least probability $1- O(\frac{1}{d^{10}})$, 
we have $G(x_0) \le \sqrt{\kappa(\mb^{-1})} d^{10.5}$.

Conditioned on high probability result of $G(x_0)$, we now use induction to prove $G(x_t) \le G(x_0)$. It trivially holds for $t=0$.
Suppose we now have $G(x) \le G(x_0)$, then by the condition in Theorem \ref{thm:init-offline} and Markov inequality, we know with probability greater than $1-\frac{1}{100\sqrt{\kappa(\mb^{-1})} d^{10.5}}$ we have: 
\begin{align*}
	\norm{\xi}_\mb \le &\frac{1}{30\sqrt{\kappa(\mb^{-1})} d^{10.5}} \cdot \norm{\frac{1}{\lambda-x^\top \bv \Sigma x} x - \mb^{-1}x}_{\mb}\\
	\le & \frac{1}{30} \cdot \norm{\frac{1}{\lambda-x^\top \bv \Sigma x} x - \mb^{-1}x}_{\mb} \min \left \{1, \frac{1}{G(x_0)} \right \} \\
	\le &\frac{1}{30} \cdot \norm{\frac{1}{\lambda-x^\top \bv \Sigma x} x - \mb^{-1}x}_{\mb} \min \left \{1, \frac{1}{G(x)} \right \} \\
	\le & \frac{\lamiBinv{1}-\lamiBinv{2}}{4}\min\left\{\sqrt{\sum_{i\geq2} \frac{\alpha_{i}^{2}}{\lamiBinv{i}}}, \sqrt{\frac{\alpha_1^2}{\lamiBinv{1}}}\right\}
\end{align*}
The last inequality uses Corollary~\ref{cor:constant_factor_corollary} with the fact that $\lamiBinv{2} \le \frac{1}{100}\lamiBinv{1}$. Therefore, we have:
We will have:
\begin{equation*}
	G(\xhat) \le  \frac{\lamiBinv{1} + 3\lamiBinv{2}}{3\lamiBinv{1}+\lamiBinv{2}} \times G(x)
	\le \frac{1}{2}G(x)
\end{equation*}
This finishes the proof of induction.

Finally, by union bound, we know with probability greater than $1- O(\frac{1}{d^{10}})$ in $T = O(\log d + \log \kappa(\mb^{-1}))$ steps, 
we have:
\begin{equation*}
	G(x_T) \le  \frac{1}{2^T} G(x_0) \le \frac{1}{\sqrt{10}}
\end{equation*}

\end{proof}
\section{Offline Eigenvector Computation}\label{sec:offline}

In this section we show how to instantiate the framework of Section \ref{framework} in order to compute an approximate top eigenvector in the offline setting. As discussed, in the offline setting we can trivially compute the Rayleigh quotient of a vector in $\nnz(\ma)$ time as we have explicit access to $\ma^\top \ma$. Consequently the bulk of our work in this section is to show how we can solve linear systems in $\mb$ efficiently in expectation, allowing us to apply Corollary \ref{cor:constant_factor_corollary} of Theorem \ref{thm:powermethod-perturb}.

In Section~\ref{sec:offline:svrg} we first show how Stochastic Variance Reduced Gradient (SVRG) \cite {johnson2013accelerating} can be adapted to solve linear systems of the form $\mb x = b$.
If we wanted, for example, to solve a linear system in a positive definite matrix like $\bv{A}^\top  \bv A$, we would optimize the objective function $f(x) = \frac{1}{2}x^\top \bv{A}^\top  \bv A x - b^\top x$. This function can be written as the sum of $n$ \emph{convex components}, $\psi_i(x) = \frac{1}{2} x^\top \left (a_i a_i^\top \right )x - \frac{1}{n}b^\top x$. In each iteration of traditional gradient descent, one computes the full gradient of $f(x_i)$ and takes a step in that direction. In stochastic gradient methods, at each iteration, a single component is sampled, and the step direction is based only on the gradient of the sampled component. Hence, we avoid a full gradient computation at each iteration, leading to runtime gains.

Unfortunately, while we have access to the rows of $\bv{A}$ and so can solve systems in $\bv A^\top \bv A$, it is less clear how to solve systems in $\bv{B} = \lambda \bv I - \bv A^\top \bv A$. To do this, we will split our function into components of the form $\psi_i(x) = \frac{1}{2} x^\top \left (w_i \bv I - a_i a_i^\top \right )x - \frac{1}{n}b^\top x$ for some set of weights $w_i$ with $\sum_{i \in [n]} w_i = \lambda$.

Importantly, $(w_i \bv I - a_i a_i^\top)$ may not be positive semidefinite. That is, we are minimizing a sum of functions which is convex, but consists of non-convex components. While recent results for minimizing such functions could be applied directly \cite{shalev2015sdca,csiba2015primal} here we show how to obtain stronger results by using a more general form of SVRG and analyzing the specific properties of our function (i.e. the variance).

Our analysis shows that we can make constant factor progress in solving linear systems in $\bv{B}$ in time $O\left (\nnz(\bv A) + \frac{d\nrank(\bv A)}{\gap^2} \right )$. If $\frac{d\nrank(\bv A)}{\gap^2} \le \nnz(\bv A)$ this gives a runtime proportional to the input size -- the best we could hope for. 
If not, we show in Section~\ref{sec:offline:acceleration}  that it is possible to \emph{accelerate} our system solver, achieving runtime  $\tilde O \left (\frac{\nnz(\bv A)^{3/4}(d\nrank(\bv A))^{1/4}}{\sqrt{\gap}} \right )$. This result uses the work of \cite{frostig2015regularizing, lin2015catalyst} on accelerated approximate proximal point algorithms. 

With our solvers in place, in
Section~\ref{sec:offline:results} we pull our results together, showing how to use these solvers in the framework of Section \ref{framework} to give faster running times for offline eigenvector computation.

\subsection{SVRG Based Solver}
\label{sec:offline:svrg}

Here we provide a sampling based algorithm for solving linear systems in $\mb$. In particular we provide an algorithm for solving the more general problem where we are given a strongly convex function that is a sum of possibly non-convex functions that obey smoothness properties. We provide a general result on bounding the progress of an algorithm that solves such a problem by non-uniform sampling in Theorem~\ref{lem:svrg-nonconv} and then in the remainder of this section we show how to bound the requisite quantities for solving linear systems in $\mb$.

\newcommand{\avgsmooth}{\overline{S}}

\begin{theorem}[SVRG for Sums of Non-Convex Functions]
\label{lem:svrg-nonconv} 
Consider a set of functions, $\{\psi_{1}, \psi_2,...\psi_n \}$, each mapping $\R^{d}\rightarrow\R$.
Let $f(x) = \sum_{i}\psi_{i}(x)$ and let $\opt{x} \eqdef \argmin_{x \in \R^{d}} f(x)$. Suppose we have a probability distribution $p$ on $[n]$, 
and that starting from some initial point $x_{0} \in \R^d$ in each iteration $k$ we pick $i_{k} \in [n]$ independently with probability $p_{i_k}$ and let 
\[
x_{k+1}
:= 
x_{k}
-\frac{\eta}{p_{i}} 
\left(\grad\psi_{i}(x_k)-\grad\psi_{i}(x_0 )\right) + \eta \grad f(x_0)
\]
for some $\eta$. If $f$ is $\mu$-strongly convex and if for all $x \in \R^d$
we have
\begin{equation}
\label{eq:svrg-nonconvex-avgsmooth}
\sum_{i \in [n]} \frac{1}{p_{i}} \norm{\grad \psi_{i}(x) - \grad\psi_{i}(\opt{x})}_{2}^{2} 
\leq 
2 \avgsmooth \left[f(x) - f(\opt{x}) \right],
\end{equation}
where $\avgsmooth$ is a variance parameter,
then for all $m \geq 1$ we have
\[
\E \left[\frac{1}{m} \sum_{k\in[m]} f(x_{k}) - f(\opt{x})\right] \leq \frac{1}{1-2\eta\bar{S}} 
\left[\frac{1}{\mu\eta m}+2\eta\avgsmooth \right] \cdot 
\left[f(x_{0})-f(\opt{x})\right]
\]
Consequently, if we pick $\eta$ to be a sufficiently small multiple of $1/\bar{S}$
then when $m = O(\overline{S} / \mu)$ we can decrease the error by
a constant multiplicative factor in expectation.
\end{theorem}

\begin{proof}
We first note that $\E_{i_{k}}[x_{k+1} - x_{k}] = \eta \grad f(x_{k})$. This is, in each iteration, in expectation, we make a step in the direction of the gradient. Using this fact we have:
\begin{align*}
\E_{i_{k}} \norm{x_{k+1} - \opt{x}}_{2}^{2} &= \E_{i_{k}} \norm{(x_{k+1} -x_k) + (x_k - \opt{x}) }_{2}^{2}  \\
&= \norm{x_k - \opt{x} }_{2}^{2} - 2 \E_{i_{k}}(x_{k+1} -x_k)^\top(x_k - \opt{x}) + \E_{i_{k}}\norm{x_{k+1} - x_k }_{2}^{2} \\
&=\norm{x_k - \opt{x} }_{2}^{2} - 2 \eta \grad f(x_k)^{\top} \left(x_{k} - \opt{x} \right)\\
&+ \sum_{i \in [n]} \eta^{2} p_{i} \normFull{
	\frac{1}{p_{i}} 
	\left(\grad \psi_{i}(x_{k}) - \grad\psi_{i}(x_{0}) \right)
	+ \grad f(x_{0})}_{2}^{2}
\end{align*}
We now apply the fact that $\norm{x + y}_2^2 \leq 2\norm{x}_2^2 + 2\norm{y}_2^2$ to give:
\begin{align*}
\sum_{i \in [n]}
& p_i \normFull{\frac{1}{p_i} 
\left(\grad \psi_i (x_k) - \grad \psi_i (x_0)\right) + \grad f(x_0)
}_2^2
\\
&
\leq 
\sum_{i \in [n]}
2 p_i \normFull{\frac{1}{p_i} 
	\left(\grad \psi_i (x_k) - \grad \psi_i (\opt{x}) \right)
}_2^2
+
\sum_{i \in [n]}
2 p_i \normFull{\frac{1}{p_i} 
	\left(\grad \psi_i (x_0) - \grad \psi_i(\opt{x}) \right) - \grad f(x_0)
}_2^2.
\end{align*}
Then, using that $\grad f(\opt{x}) = 0$ by optimality, that $\E\norm{x-\E x}_{2}^{2} \leq \E\norm{x}_{2}^{2}$, and  \eqref{eq:svrg-nonconvex-avgsmooth} we have:
\begin{align*}
\sum_{i \in [n]}
& p_i \normFull{\frac{1}{p_i} 
\left(\grad \psi_i (x_k) - \grad \psi_i (x_0)\right) + \grad f(x_0)
}_2^2
\\
&
\le
\sum_{i \in [n]}
\frac{2}{p_i}
 \normFull{\grad \psi_i (x_k) - \grad \psi_i (\opt{x})
}_2^2
+
\sum_{i \in [n]}
2 p_i \normFull{\frac{1}{p_i} 
	\left(\grad \psi_i (x_0) - \grad \psi_i(\opt{x})) - (\grad f(x_0) - \grad f(\opt{x})\right)
}_2^2
\\
&
\leq
\sum_{i \in [n]}
\frac{2}{p_i}
\normFull{\grad \psi_i (x_k) - \grad \psi_i (\opt{x})
}_2^2
+
\sum_{i \in [n]}
2 p_i \normFull{\frac{1}{p_i} 
	\grad \psi_i (x_0) - \grad \psi_i(\opt{x}))
}_2^2
\\
&
\leq 4 \avgsmooth
\left[ f(x_{k})-f(\opt{x}) + f(x_{0}) - f(\opt{x}) \right]
\end{align*}

Since $f(\opt{x}) - f(x_k) \geq \grad f(x_k)^\top (\opt{x} - x_k)$ by the convexity of $f$, these inequalities imply
\begin{align*}
\E_{i_{k}}\norm{x_{k+1} - \opt{x}}_{2}^{2} 
& 
\leq\norm{x_{k} - \opt{x}}_{2}^{2} 
- 2 \eta \left[f(x_{k}) - f(\opt{x})\right] 
+ 4 \eta^{2} \avgsmooth \left[f(x_{k}) - f(\opt{x}) + f(x_{0}) - f(\opt{x})\right] 
\\
& 
= \norm{x_{k} - \opt{x}}_{2}^{2}-2\eta(1-2\eta S)\left(f(x_{k})-f(\opt{x})\right)+4\eta^{2}\bar{S}\left(f(x_{0})-f(\opt{x})\right)
\end{align*}
Rearranging, we have:
\begin{align*}
2\eta(1-2\eta S)\left(f(x_{k})-f(\opt{x})\right) \le \norm{x_{k} - \opt{x}}_{2}^{2} - \E_{i_{k}}\norm{x_{k+1} - \opt{x}}_{2}^{2} + 4\eta^{2}\bar{S}\left(f(x_{0})-f(\opt{x})\right).
\end{align*}
And summing over all iterations and taking expectations we have:
\begin{align*}
\E \left [2\eta(1-2\eta \bar{S})\sum_{k\in[m]}f(x_{k})-f(\opt{x}) \right ] \le \norm{x_0 - \opt{x}}_2^2 + 4m\eta^{2}\bar{S}\left[f(x_{0})-f(\opt{x})\right].
\end{align*}
Finally, we use that
by strong convexity, $\norm{x_{0}-\opt{x}}_{2}^{2}\leq\frac{2}{\mu}\left(f(x_{0})-f(\opt{x})\right)$ to obtain:
	\[
	\E \left [2\eta(1-2\eta \bar{S})\sum_{k\in[m]}f(x_{k})-f(\opt{x})\right ]\leq \frac{2}{\mu}\left[f(x_{0})-f(\opt{x})\right]+4m\eta^{2}\bar{S}\left[f(x_{0})-f(\opt{x})\right]
	\]
	and thus
	\[
	\E \left [\frac{1}{m}\sum_{k\in[m]}f(x_{k})-f(\opt{x}) \right ]\leq\frac{1}{1-2\eta\bar{S}}\left[\frac{1}{\mu\eta m}+2\eta\bar{S}\right]\cdot\left[f(x_{0})-f(\opt{x})\right]
	\]
	
\end{proof}

Theorem \ref{lem:svrg-nonconv} immediately yields a solver for $\bv{B}x = b$. Finding the minimum norm solution to this system is equivalent to minimizing $f(x) = \frac{1}{2} x^\top \bv{B} x - b^\top x$. If we take the common approach of applying a smoothness bound for each $\psi_i$ along with a strong convexity bound on $f(x)$ we obtain:
\begin{lemma}[Simple Variance Bound for SVRG]\label{simple_variance_bound}
Let
\begin{align*}
	\psi_i(x) \defeq \frac{1}{2} x^\top \left(\frac{\lambda\norm{a_i}_2^2}{\norm{\bv A}_F^2} \mI - a_i a_i^\top \right) x 
	- \frac{1}{n}b^\top x
\end{align*}
so we have $\sum_{i\in[n]} \psi_i(x) = f(x) = \frac{1}{2} x^\top \bv{B} x - b^\top x$.
Setting $p_i = \frac{\norm{a_i}_2^2}{\norm{\bv A}_F^2}$ for all $i$, we have 
\begin{align*}
\sum_{i \in [n]} \frac{1}{p_i} \norm{\grad \psi_i(x)-\grad \psi_i(\opt{x})}_2^2 = O \left (
\frac{\norm{\bv A}_F^4}{\lambda - \lambda_1}
\left[f(x) - f(\opt{x})\right] \right )
\end{align*}
\end{lemma}

\begin{proof}
	We first compute, for all $i \in [n]$ 
	\begin{align}\label{basic_gradient_computation}
	\grad \psi_i(x) =
	\left(\frac{\lambda \norm{a_i}_2^2}{\normFro{\ma}^2} \mI 
	- a_i a_i^\top \right) x 
	- \frac{1}{n}b.
	\end{align}
	We have that each $\psi_i$ is $\frac{\lambda \norm{a_i}_2^2}{\normFro{\ma}^2} + \norm{a_i}^2$ smooth with respect to $\norm{\cdot}_2$. Specifically,
\begin{align*}
\norm{\grad \psi_i(x)-\grad \psi_i(\opt{x})}_2 &= \norm {\left (\frac{\lambda \norm{a_i}_2^2}{\normFro{\ma}^2} \mI 
	- a_i a_i^\top \right) (x-\opt{x}) }_2 \\
	&\le \left (\frac{\lambda \norm{a_i}_2^2}{\normFro{\ma}^2} + \norm{a_i}^2 \right ) \norm{x-\opt{x}}_2.
\end{align*}

Additionally, 
$f(x)$ is $\lambda_d(\bv B) = \lambda - \lambda_1$ strongly convex so we have $\norm{x - \opt{x}}_2^2 \leq \frac{2}{\lambda - \lambda_1} \left[f(x) - f(\opt{x})\right]$ and putting all this together we have
\begin{align*}
\sum_{i \in [n]} \frac{1}{p_i} \norm{\grad \psi_i(x) - \grad \psi_i (\opt{x})}_2^2
&\leq
\sum_{i \in [n]} \frac{\norm{\bv A}_F^2}{\norm{a_i}_2^2}
\cdot
\norm{a_i}_2^4 \left(\frac{\lambda}{\norm{\bv A}_F^2} + 1\right)^2
\cdot \frac{2}{\lambda - \lambda_1}
\left[f(x) - f(\opt{x})\right]
\\
&= O \left (
\frac{\norm{\bv A}_F^4}{\lambda - \lambda_1}
\left[f(x) - f(\opt{x})\right] \right )
\\
\end{align*}
where the last step uses that $\lambda \le 2\lambda_1 \le 2\norm{\bv A}_F^2$ so $\frac{\lambda}{\norm{\bv A}_F^2} \le 2$.
\end{proof}

Assuming that $\lambda = (1+c \cdot \gap)\lambda_1$ for some constant $c$, the above bound means that we can make constant progress on our linear system by setting $m = O(\avgsmooth/\mu) = O\left(\frac{\norm{\bv A}^4_F}{(\lambda-\lambda_1)^2}\right)= O\left ( \frac{\nrank(\bv{A})^2}{\gap^2} \right )$. This dependence on stable rank matches the dependence given in \cite{shamir2015stochastic} (see discussion in Section \ref{previous_work_offline}), however we can show that it is suboptimal. 
We show to improve the bound to $O\left ( \frac{\nrank(\bv{A})}{\gap^2} \right )$ by using a better variance analysis. Instead of bounding each $\norm{\grad \psi_i(x) - \grad \psi_i (\opt{x})}_2^2$ term using the smoothness of $\psi_i$, we more carefully bound the sum of these terms.

\begin{lemma}(Improved Variance Bound for SVRG)\label{variance_lemma} For $i \in [n]$ let
\begin{align*}
	\psi_i(x) \defeq \frac{1}{2} x^\top \left(\frac{\lambda\norm{a_i}_2^2}{\norm{\bv A}_F^2} \mI - a_i a_i^\top \right) x 
	- \frac{1}{n}b^\top x
\end{align*}
so we have $\sum_{i\in[n]} \psi_i(x) = f(x) = \frac{1}{2} x^\top \bv{B} x - b^\top x$.
Setting $p_i = \frac{\norm{a_i}_2^2}{\norm{\bv A}_F^2}$ for all $i$, we have for all $x$
	\[
	\sum_{i \in [n]}
	\frac{1}{p_i} \normFull{\grad \psi_i(x) - \grad \psi_i(\opt{x})}_2^2
	\leq 
	\frac{4\lambda_1 \normFro{\ma}^2}{\lambda - \lambda_1} \cdot \left[f(x) - f(\opt{x}) \right].
	\]
\end{lemma}

\begin{proof}
Using the gradient computation in \eqref{basic_gradient_computation} we have
	\begin{align}
	\sum_{i \in [n]}
	\frac{1}{p_i} \normFull{\grad \psi_i(x) - \grad \psi_i(\opt{x})}_2^2 &= \sum_{i \in [n]}
	\frac{\normFro{\ma}^2}{\norm{a_i}_2^2}
	\normFull{\left( \frac{\lambda\norm{a_i}_2^2}{\normFro{\ma}^2} \mI 
		- a_i a_i^\top \right) (x - \opt{x})}_2^2\nonumber\\
	&=
	\sum_{i \in [n]} 
	\frac{\lambda^2 \norm{a_i}_2^2}{\normFro{\ma}^2}
	\normFull{x - \opt{x}}^2_2
	- 2 \sum_{i \in [n]} \lambda \normFull{x - \opt{x}}_{a_i a_i^\top}^2
	\nonumber\\&+ \sum_{i \in [n]} \frac{\normFro{\ma}^2}{\norm{a_i}^2} \normFull{x - \opt{x}}_{\norm{a_i}_2^2 a_i a_i^\top}^2
	\nonumber\\
	&=
	\lambda^2 \normFull{x - \opt{x}}_2^2
	- 2 \lambda \normFull{x - \opt{x}}_{\mSigma}^2
	+ \normFro{\ma}^2 \normFull{x - \opt{x}}_{\mSigma}^2.\nonumber\\
		&\le
	\lambda \normFull{x - \opt{x}}_{\bv B}^2
	+ \normFro{\ma}^2 \normFull{x - \opt{x}}_{\mSigma}^2.\label{broken_up_norms}
	\end{align}
	
	Now since 
	\[
	\mSigma \preceq \lambda_1 \mI \preceq \frac{\lambda_1}{\lambda - \lambda_1} \mb 
	\]
	we have
	\begin{align*}
	\sum_{i \in [n]}
	\frac{1}{p_i} \normFull{\grad \psi_i(x) - \grad \psi_i(\opt{x})}_2^2  &\le \left (\frac{ \lambda(\lambda-\lambda_1) +\normFro{\ma}^2\cdot \lambda_1}{\lambda - \lambda_1} \right )\normFull{x - \opt{x}}_{\mb}^2\\
	&\le \left (\frac{2\normFro{\ma}^2 \lambda_1}{\lambda - \lambda_1} \right )\normFull{x - \opt{x}}_{\mb}^2
	\end{align*}
	where in the last inequality we just coarsely bound $\lambda(\lambda-\lambda_1) \le \lambda_1\norm{\bv A}_F^2$.
	Now since $\bv{B}$ is full rank, $\bv{B}\opt{x} = b$, we can compute:
	\begin{align}\label{norm_to_function_error_conversion}
	\norm{x - \opt{x}}_{\mb}^2 = x^\top \mb x - 2b^\top x + b^\top \opt{x} = 2[f(x) - f(\opt{x})].
	\end{align}
	The result follows.
\end{proof}

Plugging the bound in Lemma \ref{variance_lemma} into Theorem \ref{lem:svrg-nonconv} we have:
\begin{theorem}(Offline SVRG-Based Solver)\label{offline_solver}
Let $\avgsmooth =\frac{2\lambda_1 \normFro{\mSigma}^2}{\lambda - \lambda_1}$, $\mu =\lambda-\lambda_1$. The iterative procedure described in Theorem \ref{lem:svrg-nonconv} with $f(x) = \frac{1}{2}x^\top \bv{B} x - b^\top x$, $\psi_i(x) = \frac{1}{2} x^\top \left(\frac{\lambda\norm{a_i}_2^2}{\norm{\bv \Sigma}_F^2} \mI - a_i a_i^\top \right) x - b^\top x$, $p_i = \frac{\norm{a_i}_2^2}{\norm{\bv \Sigma}_F^2}$, $\eta = 1/(8\avgsmooth)$ and $m$ chosen uniformly at random from $[64 \avgsmooth/\mu]$ returns a vector $x_m$ such that
\begin{align*}
\E \norm{x_m - \opt{x}}_\mb^2 \le \frac{1}{2} \norm{x_0-\opt{x}}_\mb^2.
\end{align*}
Further, assuming $\left (1 + \frac{\gap}{150} \right ) \lambda_1< \lambda \le \left (1 + \frac{\gap}{100} \right ) \lambda_1$, this procedure runs in time $O\left (\nnz(\bv A) + \frac{d\cdot \nrank(\bv{A})}{\gap^2} \right )$.
\end{theorem}
\begin{proof}
Lemma \ref{variance_lemma} tells us that \[
	\sum_{i \in [n]}
	\frac{1}{p_i} \normFull{\grad \psi_i(x) - \grad \psi_i(\opt{x})}_2^2
	\leq 
	2 \avgsmooth \left[f(x) - f(\opt{x}) \right].
	\]
Further $f(x) = \frac{1}{2}x^\top \bv{B} x - b^\top x$ is $\lambda_d(\bv B)$-strongly  convex and $\lambda_d(\bv B) = \lambda-\lambda_1 = \mu$.
Plugging this into Theorem \ref{lem:svrg-nonconv} and using \eqref{norm_to_function_error_conversion} which shows $\norm{x-\opt{x}}_\mb^2 = 2[f(x)-f(\opt{x})]$ we have, for $m$ chosen uniformly from $[64 \avgsmooth/ \mu]$: 
\begin{align*}
\E \left[\frac{1}{64 \avgsmooth/ \mu} \sum_{k\in[64 \avgsmooth/ \mu]} f(x_{k}) - f(\opt{x})\right] &\leq 4/3 \cdot 
\left[1/8+1/8 \right] \cdot 
\left[f(x_{0})-f(\opt{x})\right]\\
\E \left [f(x_m) - f(\opt{x}) \right ] &\le \frac{1}{2}\left[f(x_{0})-f(\opt{x})\right] \\
\E \norm{x_m - \opt{x}}_\mb^2 &\le \frac{1}{2}  \norm{x_0-x^{opt}}_\mb^2.
\end{align*}

The procedure requires $O\left(\nnz(\bv A)\right)$ time to initially compute $\grad f(x_0)$, along with each $p_i$ and the step size $\eta$ which depend on $\norm{\bv A}_F^2$ and the row norms of $\bv{A}$. Each iteration then just requires $O(d)$ time to compute $\grad \psi_i(\cdot)$ and perform the necessary vector operations. Since there are at most $[64 \avgsmooth/ \mu] = O\left ( \frac{\lambda_1\norm{\bv A}_F^2}{(\lambda-\lambda_1)^2} \right )$ iterations, our total runtime is
\begin{align*}
O \left (\nnz(\bv{A})  + d \cdot \frac{\lambda_1\norm{\bv A}_F^2}{(\lambda-\lambda_1)^2} \right ) = O\left (\nnz(\bv A) + \frac{d\cdot \nrank(\bv{A})}{\gap^2} \right ).
\end{align*}
Note that if our matrix is uniformly sparse - i.e. all rows have sparsity at most $d_s$, then the runtime is actually at most $O\left (\nnz(\bv A) + \frac{d_s \cdot \nrank(\bv{A})}{\gap^2} \right )$.
\end{proof}

\subsection{Accelerated Solver}
\label{sec:offline:acceleration}

Theorem \ref{offline_solver} gives a linear solver for $\bv{B}$ that makes progress in expectation and which we can plug into Theorems \ref{thm:powermethod-perturb} and \ref{thm:init-offline}. However, we first show that the runtime in Theorem \ref{offline_solver} can be accelerated in some cases. We apply a result of \cite{frostig2015regularizing}, which shows that, given a solver for a regularized version of a convex function $f(x)$, we can produce a fast solver for $f(x)$ itself. Specifically:

\begin{lemma}[Theorem 1.1 of \cite{frostig2015regularizing}]\label{acceleration_primitive}
Let $f(x)$ be a $\mu$-strongly convex function and let $\opt{x} \eqdef \argmin_{x\in\mathbb{R}^d} f(x)$. For any $\gamma > 0$ and any $x_0 \in \mathbb{R}^d$, let $f_{\gamma,x_0}(x) \eqdef f(x) + \frac{\gamma}{2} \norm{x-x_0}_2^2$. Let $\opt{x}_{\gamma,x_0}  \eqdef \argmin_{x\in\mathbb{R}^d} f_{\gamma,x_0} (x)$.

Suppose that, for all $x_0 \in \mathbb{R}^d$, $c > 0$, $\gamma > 0$, we can compute a point $x_c$ such that
\begin{align*}
\E f_{\gamma,x_0}(x_c) - f_{\gamma,x_0}(\opt{x}_{\gamma,x_0} ) \le \frac{1}{c} \left [f_{\gamma,x_0} - f_{\gamma,x_0}(\opt{x}_{\gamma,x_0} ) \right ]
\end{align*}

in time $\mathcal{T}_c$. Then given any $x_0$, $c > 0$, $\gamma > 2 \mu$, we can compute $x_1$ such that
\begin{align*}
\E f(x_1) - f(\opt{x}) \le \frac{1}{c} \left [f(x_0) - f(\opt{x}) \right ]
\end{align*}
in time $O\left(\mathcal{T}_{4\left (\frac{2\gamma + \mu}{\mu} \right )^{3/2}} \sqrt{\lceil \gamma/\mu \rceil} \log c \right ).$
\end{lemma}

We first give a new variance bound on solving systems in $\bv{B}$ when a regularizer is used. The proof of this bound is very close to the proof given for the unregularized problem in Lemma \ref{variance_lemma}.
\begin{lemma}\label{regularized_variance_lemma} For $i \in [n]$ let
\begin{align*}
	\psi_i(x) \defeq \frac{1}{2} x^\top \left(\frac{\lambda\norm{a_i}_2^2}{\norm{\bv A}_F^2} \mI - a_i a_i^\top \right) x 
	- \frac{1}{n}b^\top x + \frac{\gamma\norm{a_i}_2^2}{2\norm{\bv A}_F^2} \norm{x - x_0}_2^2
\end{align*}
so we have $\sum_{i\in[n]} \psi_i(x) = f_{\gamma,x_0} (x) = \frac{1}{2} x^\top \bv{B} x - b^\top x + \frac{\gamma}{2}\norm{x - x_0}_2^2$.
Setting $p_i = \frac{\norm{a_i}_2^2}{\norm{\bv A}_F^2}$ for all $i$, we have for all $x$
	\[
	\sum_{i \in [n]}
	\frac{1}{p_i} \normFull{\grad \psi_i(x) - \grad \psi_i(\opt{x}_{\gamma,x_0} )}_2^2
	\leq 
	\left (  \frac{\gamma^2 + 12\lambda_1\norm{\bv A}_F^2}{\lambda-\lambda_1+\gamma}\right ) \left [f_{\gamma,x_0}(x)-f_{\gamma,x_0}(\opt{x}_{\gamma,x_0}) \right ]
	\]
\end{lemma}

\begin{proof}
	We have for all $i \in [n]$ 
	\begin{align}\label{regularized_gradient_computation}
	\grad \psi_i(x) =
	\left(\frac{\lambda \norm{a_i}_2^2}{\normFro{\ma}^2} \mI 
	- a_i a_i^\top \right) x 
	- \frac{1}{n}b + \frac{\gamma\norm{a_i}_2^2}{2\norm{\bv A}_F^2} (x-2x_0)
	\end{align}
Plugging this in we have:
	\begin{align*}
	\sum_{i \in [n]}
	\frac{1}{p_i} \normFull{\grad \psi_i(x) - \grad \psi_i(\opt{x}_{\gamma,x_0} )}_2^2 &= \sum_{i \in [n]}
	\frac{\normFro{\ma}^2}{\norm{a_i}_2^2}
	\normFull{\left( \frac{\lambda\norm{a_i}_2^2}{\normFro{\ma}^2} \mI 
		- a_i a_i^\top \right) (x - \opt{x}_{\gamma,x_0} ) + \frac{\gamma\norm{a_i}_2^2}{2\norm{\bv A}_F^2} (x - \opt{x}_{\gamma,x_0} )}_2^2
	\end{align*}
	For simplicity we now just use the fact that  $\norm{x+y}_2^2 \le 2\norm{x}_2^2 + 2\norm{y}_2^2$ and apply our bound from equation \eqref{broken_up_norms} to obtain:
	\begin{align*}
	\sum_{i \in [n]}
	\frac{1}{p_i} \normFull{\grad \psi_i(x) - \grad \psi_i(\opt{x}_{\gamma,x_0} )}_2^2  &\le 
	2\lambda^2 \normFull{x - \opt{x}_{\gamma,x_0}}_2^2
	- 4 \lambda \normFull{x - \opt{x}_{\gamma,x_0}}_{\mSigma}^2
	+ 2\normFro{\mSigma}^2 \normFull{x - \opt{x}_{\gamma,x_0}}_{\mSigma}^2 \\&+ 2\sum_{i \in [n]}\frac{\norm{a_i}_2^2}{\normFro{\ma}^2} \frac{\gamma^2}{4} \norm{x-\opt{x}_{\gamma,x_0} }_2^2\\
	&\le 
	 \left (2\lambda^2 +\gamma^2/2 + 2\lambda_1\norm{\bv A}_F^2 - 4\lambda_1\lambda \right ) \normFull{x - \opt{x}_{\gamma,x_0} }_{2}^2\\
	&\le 
	 \left (\gamma^2/2 + 6\lambda_1\norm{\bv A}_F^2 \right ) \normFull{x - \opt{x}_{\gamma,x_0} }_{2}^2
	\end{align*}
	
	Now, $f_{\gamma,x_0}(\cdot)$ is $\lambda-\lambda_1+\gamma$ strongly convex, so 
	\begin{align*}
	\normFull{x - \opt{x}_{\gamma,x_0} }_{2}^2 \le \frac{2}{\lambda-\lambda_1+\gamma} \left [f_{\gamma,x_0}(x)-f_{\gamma,x_0}(\opt{x}_{\gamma,x_0}) \right ].
	\end{align*}
	So overall we have:
	\begin{align*}
	\sum_{i \in [n]} \frac{1}{p_i} \normFull{\grad \psi_i(x) - \grad \psi_i(\opt{x}_{\gamma,x_0} )}_2^2  &\le \left (  \frac{\gamma^2 + 12\lambda_1\norm{\bv A}_F^2}{\lambda-\lambda_1+\gamma}\right ) \left [f_{\gamma,x_0}(x)-f_{\gamma,x_0}(\opt{x}_{\gamma,x_0}) \right ]
	\end{align*}
\end{proof}

We can now use this variance bound to obtain an accelerated solver for $\bv{B}$. We assume $\nnz(\bv{A}) \le \frac{d\nrank(\bv A)}{\gap^2}$, as otherwise, the unaccelerated solver in Theorem \ref{offline_solver} runs in $O(\nnz(\bv{A}))$ time and cannot be accelerated further.
\begin{theorem}[Accelerated SVRG-Based Solver]\label{accelerated_offline_solver}
Assuming $\left (1 + \frac{\gap}{150} \right ) \lambda_1< \lambda \le \left (1 + \frac{\gap}{100} \right ) \lambda_1$ and $nnz(\bv{A}) \le \frac{d\nrank(\bv A)}{\gap^2}$, applying the iterative procedure described in Theorem \ref{lem:svrg-nonconv} along with the acceleration given by Lemma \ref{acceleration_primitive} gives a solver that returns $x$ with
\begin{align*}
\E \norm{x - \opt{x}}_\mb^2 \le \frac{1}{2} \norm{x_0-\opt{x}}_\mb^2.
\end{align*}
in time $O\left ( \frac{\nnz(\bv A)^{3/4} (d \nrank(\bv A))^{1/4}}{\sqrt{\gap}} \cdot \log \left (\frac{d}{\gap}\right)\right)$.
\end{theorem}
\begin{proof}
Following Theorem \ref{offline_solver}, the variance bound of Lemma \ref{regularized_variance_lemma} means that we can make constant progress in minimizing $f_{\gamma,x_0}(x)$ in $O\left (\nnz(\bv{A}) + d m \right)$ time where $m = O \left ( \frac{\gamma^2 + 12\lambda_1\norm{\bv\Sigma}_F^2}{(\lambda-\lambda_1+\gamma)^2}\right )$. So, for $\gamma \ge 2(\lambda-\lambda_1)$ we can make $4\left (\frac{2\gamma + (\lambda-\lambda_1)}{\lambda-\lambda_1} \right )^{3/2}$ progress, as required by Lemma \ref{acceleration_primitive} in time $O\left (\left (\nnz(\bv A) + dm \right) \cdot \log \left (\frac{\gamma}{\lambda-\lambda_1}\right) \right )$ time. Hence by Lemma \ref{acceleration_primitive} we can make constant factor expected progress in minimizing $f(x)$ in time:
\begin{align*}
O\left ( \left (\nnz(\bv A) + d\frac{\gamma^2 + 12\lambda_1\norm{\bv A}_F^2}{(\lambda-\lambda_1+\gamma)^2} \right ) \log \left (\frac{\gamma}{\lambda-\lambda_1}\right) \sqrt {\frac{\gamma}{\lambda-\lambda_1}}\right)
\end{align*}

By our assumption, we have $\nnz(\bv{A}) \le \frac{d\nrank(\bv A)}{\gap^2} = \frac{d \lambda_1 \norm{\bv A}_F^2}{(\lambda-\lambda_1)^2}$. So, if we let $\gamma = \Theta \left (\sqrt{\frac{d\lambda_1 \norm{\ma}_F^2}{\nnz(\bv A)}}\right)$ then using a sufficiently large constant, we have $\gamma \ge 2(\lambda-\lambda_1)$. We have $\frac{\gamma}{\lambda-\lambda_1} = \Theta \left (\sqrt{\frac{d\lambda_1 \norm{\bv A}_F^2}{\nnz(\bv A)\lambda_1^2 \gap^2}}\right) = \Theta \left (\sqrt{\frac{d\nrank(\bv A)}{\nnz(\bv A)\gap^2}}\right) $ and can make constant expected progress in minimizing $f(x)$ in time:
\begin{align*}
O\left ( \frac{\nnz(\bv A)^{3/4} (d \nrank(\bv A))^{1/4}}{\sqrt{\gap}} \cdot \log \left (\frac{d}{\gap}\right)\right).
\end{align*}
\end{proof}

\subsection{Shifted-and-Inverted Power Method}
\label{sec:offline:results}

Finally, we are able to combine the solvers from Sections \ref{sec:offline:svrg} and \ref{sec:offline:acceleration} with the framework of Section \ref{framework} to obtain faster algorithms for top eigenvector computation.
\begin{theorem}[Shifted-and-Inverted Power Method With SVRG]\label{main_offline_theorem}
Let $\bv{B} = \lambda \bv{I} - \bv{A}^\top \bv{A}$ for $\left ( 1+\frac{\gap}{150}\right) \lambda_1 \le \lambda \le \left(1+ \frac{\gap}{100} \right ) \lambda_1$ and let $x_0 \sim \mathcal{N}(0,\bv I)$ be a random initial vector. Running the inverted power method on $\bv{B}$ initialized with $x_0$, using the SVRG solver from Theorem \ref{offline_solver} to approximately apply $\bv{B}^{-1}$ at each step, returns $x$ such that with probability $1-O\left (\frac{1}{d^{10}}\right)$, $x^\top \bv{\Sigma}x \ge (1-\epsilon) \lambda_1$ in total time $$O \left (\left(\nnz(\bv A) + \frac{d \nrank(\bv A)}{\gap^2} \right )\cdot \left (\log^2\left(\frac{d}{\gap}\right) + \log\left(\frac{1}{\epsilon}\right) \right ) \right ).$$
\end{theorem}
Note that by instantiating the above theorem with $\epsilon' = \epsilon\cdot \gap$, and applying Lemma \ref{lem:ray_to_evec} we can find a unit vector $x$ such that $| v_1^\top x | \ge 1-\epsilon$ in the same asymptotic running time (an extra $\log(1/\gap)$ term is absorbed into the $\log^2(d/\gap)$ term).
\begin{proof}
By Theorem \ref{thm:init-offline}, if we start with $x_0 \sim \mathcal{N}(0,\bv I)$ we can run $O\left (\log \left (\frac{d}{\gap} \right ) \right )$ iterations of the inverted power method, to obtain $x_1$ with $G(x_1) \le \frac{1}{\sqrt{10}}$ with probability $1-O\left (\frac{1}{d^{10}}\right)$. Each iteration requires applying an linear solver that decreases initial error in expectation by a factor of $\frac{1}{\poly(d,1/\gap)}$. Such a solver is given by applying the solver in Theorem \ref{offline_solver} $O\left (\log \left (\frac{d}{\gap} \right ) \right )$ times, decreasing error by a constant factor in expectation each time. So overall in order to find $x_1$ with $G(x_1) \le \frac{1}{\sqrt{10}}$, we require time $O \left (\left(\nnz(\bv A) + \frac{d \nrank(\bv A)}{\gap^2} \right )\cdot\log^2\left(\frac{d}{\gap}\right)\right )$.

After this initial `burn-in' period we can apply Corollary \ref{cor:constant_factor_corollary} of Theorem \ref{thm:powermethod-perturb}, which shows that running a single iteration of the inverted power method will decrease $G(x)$ by a constant factor in expectation. In such an iteration, we only need to use a solver that decreases initial error by a constant factor in expectation. So we can perform each inverted power iteration in this stage in time $O \left (\nnz(\bv A) + \frac{d \nrank(\bv A)}{\gap^2} \right ).$

With $O\left (\log \left (\frac{d}{\epsilon}\right)\right)$ iterations, we can obtain $x$ with $\E \left [G(x)^2 \right ] = O\left (\frac{\epsilon}{d^{10}}\right)$ So by Markov's inequality, we have $G(x)^2 = O(\epsilon)$, giving us $x^T \bv{\Sigma} x \ge (1-O(\epsilon))\lambda_1$ by Lemma \ref{lem:rayquot-potential}. Union bounding over both stages gives us failure probability $O\left ( \frac{1}{d^{10}}\right )$, and adding the runtimes from the two stages gives us the final result. Note that the second stage requires $O\left (\log \left (\frac{d}{\epsilon}\right)\right) = O(\log d + \log (1/\epsilon))$ iterations to achieve the high probability bound. However, the $O(\log d)$ term is smaller than the $O\left (\log^2 \left (\frac{d}{\gap} \right ) \right )$ term, so is absorbed into the asymptotic notation.

\end{proof}

We can apply an identical analysis using the accelerated solver from Theorem \ref{accelerated_offline_solver}, obtaining the following runtime which beats Theorem \ref{main_offline_theorem} whenever $\nnz(\bv A) \le \frac{d\nrank(\bv A)}{\gap^2}$:
\begin{theorem}[Shifted-and-Inverted Power Method Using Accelerated SVRG]\label{accelerated_offline_theorem}
Let $\bv{B} = \lambda \bv{I} - \bv{A}^\top \bv{A}$ for $\left ( 1+\frac{\gap}{150}\right) \lambda_1 \le \lambda \le \left(1+ \frac{\gap}{100} \right ) \lambda_1$ and let $x_0 \sim \mathcal{N}(0,\bv I)$ be a random initial vector. Assume that $\nnz(\bv A) \le \frac{d\nrank(\bv A)}{\gap^2}$. Running the inverted power method on $\bv{B}$ initialized with $x_0$, using the accelerated SVRG solver from Theorem \ref{accelerated_offline_solver} to approximately apply $\bv{B}^{-1}$ at each step, returns $x$ such that with probability $1-O\left (\frac{1}{d^{10}}\right)$, $| v_1^\top x | \ge 1-\epsilon$ in total time $$O \left (\left(\frac{\nnz(\bv A)^{3/4}(d \nrank(\bv A))^{1/4}}{\sqrt{\gap}} \right )\cdot \left (\log^3\left(\frac{d}{\gap}\right) + \log\left(\frac{d}{\gap}\right)\log\left(\frac{1}{\epsilon}\right) \right ) \right ).$$
\end{theorem}

\section{Online Eigenvector Computation}\label{sec:online}

Here we show how to apply the shifted-and-inverted power method framework of Section \ref{framework} to the online setting. This setting is more difficult than the offline case. As there is no canonical matrix $\bv{A}$, and we only have access to the distribution $\dist$ through samples, in order to apply Theorem \ref{thm:powermethod-perturb}
we must show
 how to both estimate the Rayleigh quotient (Section~\ref{sec:online:rayleigh_quotient}) as well as solve the requisite linear systems in expectation (Section~\ref{sec:online:system-solver}).

After laying this ground work, our main result is  given in Section \ref{sec:online:result}. 
Ultimately, the results in this section allow us to achieve more efficient algorithms for computing the top eigenvector in the statistical setting as well as improve upon the previous best known sample complexity for top eigenvector computation. As we show in Section~\ref{sec:lower} the bounds we provide in this section are in fact tight for general distributions.

\subsection{Estimating the Rayleigh Quotient}
\label{sec:online:rayleigh_quotient}

Here we show how to estimate the Rayleigh quotient of a vector with respect to $\mSigma$. Our analysis is standard -- we first approximate the Rayleigh quotient by its empirical value on a batch of $k$ samples and prove using Chebyshev's inequality that the error on this sample is small with constant probability. We then repeat this procedure $O(\log (1/p))$ times and output the median. By a Chernoff bound this yields a good estimate with probability $1 - p$. The formal statement of this result and its proof comprise the remainder of this subsection. 

\begin{theorem}[Online Rayleigh Quotient Estimation]\label{online_rayleigh_estimation}
Given $\epsilon \in (0,1]$, $p \in [0, 1]$, and unit vector $x$ set $k = \lceil 4\nvar(\dist) \epsilon^{-2} \rceil$ and $m = O(\log(1/p))$.
For all $i \in [k]$ and $j \in [m]$ let $a_{i}^{(j)}$ be drawn independently from $\dist$ and set $R_{i,j} = x^\top a_i^{(j)} (a_i^{(j)})^\top x$ and $R_{j} = \frac{1}{k} \sum_{i \in [k]} R_{i,j}$. If we let $z$ be median value of the $R_j$ then with probability $1 - p$ we have 
$
\left|z - x^\top \mSigma x\right| \leq \epsilon \lambda_1 
$.
\end{theorem}

\begin{proof}
\begin{align*}
\variance_{a \sim \dist} (x^\top a a^\top x)
&=
\E_{a \sim \dist} (x^\top a a^\top x)^2 - (\E_{a \sim \dist} x^\top a a^\top x)^2
\\
&\leq 
\E_{a \sim \dist} \norm{a}_2^2 x^\top a a^\top x - (x^\top \mSigma x)^2
\\
&\leq \normFull{\E_{a \sim \dist} \norm{a}_2^2 a a^\top}_2 =\nvar(\dist) \lambda_1^2
\end{align*}
Consequently, $\variance(R_{i,j}) \leq \nvar(\dist) \lambda_1^2$, and since each of the $a_{i}^{(j)}$ were drawn independently this implies that we have that $\variance(R_j) \leq \nvar(\dist) \lambda_1^2 / k$. Therefore, by Chebyshev's inequality 
\[
\Pr\left[\left|R_j - \E[R_j]\right| \geq 2 \sqrt{\frac{\nvar(\dist) \lambda_1^2}{k}}\right]
\leq \frac{1}{4}.
\]
Since $\E[R_j] = x^\top \mSigma x$ and since we defined $k$ appropriately this implies that
\begin{equation}
\label{eq:rayleighlemma:1}
\Pr\left[\left|R_j - x^\top \mSigma x \right| \geq \epsilon \lambda_1 \right]
\leq \frac{1}{4}.
\end{equation}
The median $z$ satisfies $|z - x^\top \mSigma x| \leq \epsilon$ as more than half of the $R_j$ satisfy $|R_j - x^\top \mSigma x| \leq \epsilon$. This happens with probability $1 - p$ by Chernoff bound, our choice of $m$ and \eqref{eq:rayleighlemma:1}.
\end{proof}

\subsection{Solving the Linear system}
\label{sec:online:system-solver}

Here we show how to solve linear systems in $\mb = \lambda \mI - \mSigma$ in the streaming setting. We follow the general strategy of the offline algorithms in Section~\ref{sec:offline}, replacing traditional SVRG with the streaming SVRG algorithm of \cite{frostig2014competing}. Similarly to the offline case we minimize $f(x) = \frac{1}{2} x^\top \bv B x - b^\top x$ and define for all $a \in \supp(\dist)$,
\begin{align}\label{distribution_obj_function}
\psi_{a}(x)
	\defeq
\frac{1}{2} x^{\top} (\lambda \mI - aa^{\top})x - b^{\top}x.
\end{align}
insuring that 
$
f(x) = \E_{a \sim \dist} \psi_{a}(x).
$.

The performance of streaming SVRG \cite{frostig2014competing} is governed by three regularity parameters. As in the offline case, we use the fact that $f(\cdot)$ is $\mu$-strongly convexity for $\mu = \lambda - \lambda_1$ and we require a smoothness parameter, denoted $\avgsmooth$, that satisfies:
\begin{equation}
\label{eq:smoothness}
\forall x \in \R^d \enspace : \enspace
\E_{a \sim \dist}\norm{\grad\psi_{a}(x)-\grad\psi_{a}(\opt{x})}_2^2
\leq 2 \avgsmooth \left [f(x) - f(\opt{x})\right] ~.
\end{equation}
Furthermore, we require an upper bound the variance, denoted $\sigma^2$, that satisfies:
\begin{equation}
\label{eq:variance}
\E_{a \sim \dist}    \frac{1}{2}
\normFull{\grad \psi_{a} (\opt{x})}_{
	\left(\hess f (\opt{x})\right)^{-1}}^2
\leq \sigma^2
~.
\end{equation}
With the following two lemmas we bound these parameters. 

\begin{lemma} [Streaming Smoothness]
\label{lem:online:smooth}
The smoothness parameter
$
\avgsmooth \defeq \lambda + \frac{\nvar(\dist)\lambda_1^2}{\lambda - \lambda_1}
$
satisfies \eqref{eq:smoothness}.
\end{lemma}

\begin{proof} Our proof is similar to the one for Lemma~\ref{simple_variance_bound}.
\begin{align*}
\E_{a \sim \dist}
\normFull{\grad\psi_{a}(x) - \grad \psi_{a} (\opt{x})}_2^2 &= \E_{a \sim \dist} \normFull{(\lambda \mI - aa^\top) (x - \opt{x})}_2^2\\
	&=
	\lambda^2 \norm{x - \opt{x}}_2^2 - 2\lambda \E_{a \sim \dist} \norm{x - \opt{x}}_{aa^\top}^2 + \E_{a \sim \dist} \norm{aa^\top (x - \opt{x})}_2^2\\
	&\le
	\lambda^2 \norm{x - \opt{x}}_2^2 - 2\lambda \norm{x - \opt{x}}_{\bv \Sigma}^2 + \normFull{\E_{a \sim \dist} \norm{a}_2^2 a a^\top}_2 \cdot \norm{x - \opt{x}}_2^2\\
	&\leq  \lambda \normFull{x - \opt{x}}_\mb^2 + \nvar(\dist)\lambda_1^2\norm{x - \opt{x}}_2^2.
\end{align*}
Since $f$ is $\lambda - \lambda_1$-strongly convex, $\norm{x - \opt{x}}_2^2 \leq \frac{2}{\lambda - \lambda_1} [f(x) - f(\opt{x})]$. Furthermore, since direct calculation reveals, $2[f(x) - f(\opt{x})] = \norm{x - \opt{x}}_{\mb}^2$, the result follows.
\end{proof}

\begin{lemma}[Streaming Variance] 
\label{lem:online:streamvar}
The variance parameter $\sigma^2 \defeq
\frac{\nvar(\dist)\lambda_1^2}{\lambda - \lambda_1} 
\normFull{\opt{x}}_2^2$ 
satisfies \eqref{eq:variance}.
\end{lemma}

\begin{proof}
We have
\begin{align*}
\E_{a \sim \dist}  \frac{1}{2}
\normFull{\grad \psi_{a} (\opt{x})}_{
	\left(\hess f (\opt{x})\right)^{-1}}^2
	 &=
\E_{a \sim \dist} \frac{1}{2} \normFull{\left(\lambda \mI 
	- aa^{\top}\right) \opt{x} - b}_{\mb^{-1}}^{2}
	\\
	&=
\E_{a \sim \dist} \frac{1}{2} \normFull{\left(\lambda \mI 
	- aa^{\top}\right) \opt{x} - \mb \opt{x}}_{\mb^{-1}}^{2}
	\\
	&=
\E_{a \sim \dist} \frac{1}{2} \normFull{\left(\mSigma
	- aa^{\top}\right) \opt{x}}_{\mb^{-1}}^{2}.
\end{align*}
Applying $\E \norm{a -\E a}_{2}^{2} = \E\norm a_{2}^{2}-\norm{\E a}_{2}^{2}$ gives:
\[
\E_{a \sim \dist} \normFull{\left(\mSigma
	- aa^{\top}\right) \opt{x}}_{\mb^{-1}}^{2}
 = \E_{a \sim \dist} \normFull{\opt{x}}_{a a^\top \mb^{-1} a a^\top}^2
-  \normFull{\opt{x}}_{\mSigma \mb^{-1} \mSigma}^2 \le \E_{a \sim \dist}  \normFull{\opt{x}}_{a a^\top \mb^{-1} a a^\top}^2.
\]
Furthermore, since $\mb^{-1} \preceq \frac{1}{\lambda - \lambda_1} \mI$ we have 
\[
\E_{a \sim \dist}  a a^{\top} \mb^{-1} a a^\top
\preceq
\frac{1}{\lambda - \lambda_1} \E_{a \sim \dist}  (a a^\top)^2
\preceq
\left(
\frac{\normFull{\E_{a \sim \dist} (a a ^\top)^2}_2}{\lambda - \lambda_1}\right) \mI
=
\left(\frac{\nvar(\dist)\lambda_1^2}{\lambda - \lambda_1}\right) \mI ~.
\]
Combining these three equations yields the result.
\end{proof}

With the regularity parameters bounded  we can apply the streaming SVRG algorithm of \cite{frostig2014competing} to solve systems in $\bv{B}$. We encapsulate the core iterative step of Algorithm $1$ of \cite{frostig2014competing} as follows:


\begin{definition}[Streaming SVRG Step]\label{svrg_step_def} 
Given $x_0 \in \R^d$ and $\eta, k, m > 0$ we define a \emph{streaming SVRG step}, $x = \ssvrgstep(x_0, \eta, k, m)$ as follows. First we take $k$ samples $a_1, ..., a_k$ from $\dist$ and set $g = \frac{1}{k} \sum_{i \in [k]} \psi_{a_i}$ where $\psi_{a_i}$ is as defined in \eqref{distribution_obj_function}. Then for $\wt{m}$ chosen uniformly at random from $\{1, ..., m\}$ we draw $\tilde m$ additional samples $\wt{a}_1, ..., \wt{a}_{\wt{m}}$ from $\dist$. For $t = 0, ... , \wt{m} - 1$ we let
\[
x_{t+1}
	:=
x_t -
	\frac{\eta}{L}
\left(
	\grad \psi_{\wt{a}_{t}}(x_t)
	- \grad \psi_{\wt{a}_{t}}(x_0)
	+ \grad g(x_0)
\right)
	\]
and return $x_{\wt{m}}$ as the output.
\end{definition}

The accuracy of the above iterative step is proven in Theorem 4.1 of \cite{frostig2014competing}, which we include, using our notation below:

\begin{theorem} [Theorem 4.1 of \cite{frostig2014competing} \footnote{Note that 
		Theorem~4.1 in \cite{frostig2014competing} has an additional parameter of $\alpha$, which bounds the Hessian of $f(\opt{x})$ in comparison to the Hessian everywhere else. In our setting this parameter is $1$ as $\hess f(y) = \hess f(z)$ for all $y$ and $z$.}]\label{streaming_svrg_perf_bound}
Let $f(x) = \E_{a\sim \dist} \psi_a(x)$ and let $\mu$, $\avgsmooth$, $\sigma^2$ be the strong convexity, smoothness, and variance bounds for $f(x)$. Then for any distribution over $x_0$ we have that 
$x := \ssvrgstep(x_0, \eta, k, m)$ has
$\E[f(x) - f(\opt{x})]$ upper bounded by
\[
	\frac{1}{1 - 4\eta}
\left[
	\left(\frac{\avgsmooth}{\mu m\eta} + 4 \eta \right)
	\left [\E f(x_0) - f(\opt{x}) \right ]
+ \frac{1 + 2\eta}{k}
	\left(
	\sqrt{\frac{\avgsmooth}{\mu} \cdot \left [\E f(x_0) - f(\opt{x})\right ]}
	+ \sigma
	\right)^2
	\right].
\]
\end{theorem}

Using Theorem~\ref{streaming_svrg_perf_bound} we can immediately obtain the following guarantee for solve system in $\bv{B}$:

\begin{corollary}[Streaming SVRG Solver - With Initial Point]\label{streaming_solverOLD}
Let $\mu =\lambda-\lambda_1$, $\avgsmooth =\lambda + \frac{\nvar(\dist)\lambda_1^2}{\lambda-\lambda_1}$, and $\sigma^2 = \frac{\nvar(\dist)\lambda_1^2}{\lambda-\lambda_1} \norm{\opt{x}}_2^2$. Let $c_2,c_3 \in (0,1)$ be any constants and set $\eta = \frac{c_2}{8}$, $m = \left [\frac{\avgsmooth}{\mu c_2^2}\right]$, and $k = \max \left \{ \left [\frac{\avgsmooth}{\mu c_2}\right], \left [\frac{\nvar(\dist)\lambda_1^2}{(\lambda-\lambda_1)^2c_3} \right] \right \}$. If to solve $\bv{B}x = b$ for unit vector $b$ with initial point $x_0$, we use the iterative procedure described in Definition \ref{svrg_step_def} to compute $x = \ssvrgstep(x_0, \eta, k, m)$ then:
\[
\E \norm{x - \opt{x}}_{\mb}^2
\leq 22c_2 \cdot \norm{x_0-\opt{x}}_\mb^2 + 10 c_3\lambda_1(\bv{B}^{-1}).
\]
Further, the procedure requires $O \left ( \frac{\nvar(\dist)}{\gap^2} \left [\frac{1}{c_2^2} + \frac{1}{c_3} \right ] \right )$ samples from $\dist$.
\end{corollary}

\begin{proof}
Using the inequality $(x + y)^2 \leq 2x^2 + 2y^2$ we have that
\[
\left(\sqrt{\frac{\avgsmooth}{\mu} \cdot \E[f(x_0) - f(\opt{x})]} + \sigma \right)^2
\leq \frac{2 \avgsmooth}{\mu} \cdot \E[f(x_0) - f(\opt{x})]
+ 2\sigma^2
\]
Additionally, since $b$ is a unit vector, we know that $\norm{\opt{x}}_2^2 = \norm{\mb^{-1} b}_2^2 \le \frac{1}{(\lambda-\lambda_1)^2}$.
Using equation \eqref{norm_to_function_error_conversion}, i.e. that $\norm{x-\opt{x}}_\bv{B}^2 = 2 [f(x)-f(\opt{x})]$ for all $x$,  we have by Theorem \ref{streaming_svrg_perf_bound}:
\begin{align*}
\E \norm{x-\opt{x}}_\bv{B}^2 
&\leq
\frac{1}{1 - c_2/2} \left[
\left(8c_2 + \frac{c_2}{2} + \frac{4+c_2}{2}\cdot c_2 \right)\cdot \norm{x_0-\opt{x}}_\mb^2
+ \frac{4+c_2}{4k}\cdot \frac{\nvar(\dist)\lambda_1^2}{(\lambda-\lambda_1)^3}
\right]
\\
&\le 22c_2 \cdot \norm{x_0-\opt{x}}_\mb^2 + \frac{10 c_3}{\lambda-\lambda_1}
= 22c_2 \cdot \norm{x_0-\opt{x}}_\mb^2 + 10 c_3 \lambda_1(\bv{B}^{-1}).
\end{align*}
Since $1/(\lambda - \lambda_1) = \lambda_1(\bv{B}^{-1})$ we see that $\E \norm{x-\opt{x}}_\bv{B}^2$ is as desired. All that remains is to bound the number of samples we used.

Now the number of samples used to compute $x$ is clearly at most $m + k$ Now
\[
m = \frac{\avgsmooth}{\mu c_2^2} = O  \left(\frac{\lambda}{c_2^2(\lambda-\lambda_1)} + \frac{\nvar(\dist)\lambda_1^2}{c_2^2(\lambda-\lambda_1)^2}\right ) = O\left (\frac{1}{c_2^2\gap} + \frac{\nvar(\dist)}{c_2^2\gap^2} \right ) ~.
\]. However since $\gap < 1$ and $\nvar(\dist) \ge 1$ this simplifies to $m = O\left ( \frac{\nvar}{c_2^2\gap^2}\right )$. Next to bound $k$ we can ignore the $\left [\frac{\avgsmooth}{\mu c_2}\right]$ term since this was already included in our bound of $m$ and just bound $\frac{\nvar(\dist)\lambda_1^2}{c_3(\lambda-\lambda_1)^2} = O\left ( \frac{\nvar(\dist)}{\gap^2c_3}\right )$ yielding our desired sample complexity. 
\end{proof}

Whereas in the offline case, we could ensure that our initial error $\norm{x_0 - \opt{x}}_\mb^2$ is small by simply scaling by the Rayleigh quotient (Corollary~\ref{cor:constant_factor_corollary}) in the online case estimating the Rayleigh quotient to sufficient accuracy would require too many samples. Instead, here simply show how to simply apply Corollary~\ref{streaming_solverOLD} iteratively to solve the desired linear systems to absolute accuracy without an initial point. Ultimately, due to the different error dependences in the online case this guarantee suffices and the lack of an initial point is not a bottleneck.

\begin{corollary}[Streaming SVRG Solver] \label{streaming_solver}
There is a streaming algorithm that iteratively applies the solver of Corollary~\ref{streaming_solverOLD} to solve $\bv{B}x = b$ for unit vector $b$ and returns a vector $x$ that satisfies
$
\E \norm{x - \opt{x}}_{\mb}^2
\leq 10 c\lambda_1(\bv{B}^{-1})
$
using $O \left ( \frac{\nvar(\dist)}{\gap^2 \cdot c}\right )$ samples from $\dist$.
\end{corollary}
\begin{proof}
Let $x_0 = 0$. Then $\norm{x_0-\opt{x}}_\mb^2 = \norm{\bv{B}^{-1} b}_\mb^2 \le \lambda_1(\bv{B}^{-1})$ since $b$ is a unit vector. If we apply Corollary \ref{streaming_solverOLD} with $c_2 = \frac{1}{44}$ and $c_3 = \frac{1}{20}$, then we will obtain $x_1$ with $\E \norm{x_1-\opt{x}}_\mb^2 \le \frac{1}{2}\lambda_1(\mb^{-1})$. If we then double $c_3$ and apply the solver again we obtain $x_2$ with  $\E \norm{x_1-\opt{x}}_\mb^2 \le \frac{1}{4}\lambda_1(\mb^{-1})$. Iterating in this way, after $\log(1/c)$ iterations we will have the desired guarantee: $\E \norm{x - \opt{x}}_{\mb}^2 \leq 10 c\lambda_1(\bv{B}^{-1}).$ Our total sample cost in each iteration is, by Corollary \ref{streaming_solverOLD}, $O \left ( \frac{\nvar(\dist)}{\gap^2} \left [\frac{1}{44^2} + \frac{1}{c_3} \right ] \right ).$ Since we double $c_3$ each time, the cost corresponding to the $\frac{1}{c_3}$ terms is dominated by the last iteration when we have $c_3 = O(c)$. So our overall sample cost is just:
\begin{align*}
O \left ( \frac{\nvar(\dist)}{\gap^2}\left [\frac{1}{c} + \log(1/c) \right ] \right ) = O \left ( \frac{\nvar(\dist)}{\gap^2 \cdot c}\right ).
\end{align*}
\end{proof}

\subsection{Online Shifted-and-Inverted Power Method}
\label{sec:online:result}

We now apply the results in Section~\ref{sec:online:rayleigh_quotient} and Section~\ref{sec:online:system-solver} to the shifted-and-inverted power method framework of Section~\ref{framework} to give our main result in the online setting, an algorithm that quickly refines a coarse approximation to $v_1$ into a finer approximation.

\begin{theorem}[Online Shifted-and-Inverted Power Method -- Warm Start]\label{warmstart_online_theorem}
	Let $\bv{B} = \lambda \bv{I} - \bv{A}^\top \bv{A}$ for $\left ( 1+\frac{\gap}{150}\right) \lambda_1 \le \lambda \le \left(1+ \frac{\gap}{100} \right ) \lambda_1$ and let $x_0$ be some vector with  $G(x_0)\leq \frac{1}{\sqrt{10}}$. Running the shifted-and-inverted power method on $\bv{B}$ initialized with $x_0$, using the streaming SVRG solver of Corollary \ref{streaming_solver} to approximately apply $\bv{B}^{-1}$ at each step, returns $x$ such that   $x^\top \bv{\Sigma} x \ge (1-\epsilon) \lambda_1$ with constant probability for any target $\epsilon < \gap$. The algorithm uses 
	$
	O (\frac{\nvar(\dist)}{\gap \cdot \epsilon})
	$ samples and amortized $O(d)$ time per sample.
\end{theorem}
We note that by instantiating Theorem \ref{warmstart_online_theorem}, with $\epsilon' = \epsilon\cdot \gap$, and applying Lemma \ref{lem:ray_to_evec} we can find $x$ such that $| v_1^\top x | \ge 1-\epsilon$ with constant probability in time $O \left (\frac{\nvar(\dist)}{\gap^2 \cdot \epsilon } \right ).
$
\begin{proof}
By Lemma \ref{lem:rayquot-potential} it suffices to have $G^2(x) = O(\frac{\epsilon}{\gap})$ or equivalently $G(x) = O(\sqrt{\epsilon/\gap})$. In order to succed with constant probability it suffices to have $\expec{ G(x)} = O( \sqrt{\epsilon/\gap})$ with constant probability. Since we start with $G(x_0)\leq \frac{1}{\sqrt{10}}$, we can achieve this
using $\log(\gap/\epsilon)$ iterations of the approximate shifted-and-inverted power method of Theorem \ref{thm:powermethod-perturb}. In each iteration $i$ we choose the error parameter for Theorem \ref{thm:powermethod-perturb} to be $c_1(i) = \frac{1}{\sqrt{10}}\cdot \left( \frac{1}{5}\right)^i$. Consequently, 
\begin{align*}
\expec{G(x_i)} \le \frac{3}{25} G(x_{i-1}) +  \frac{4}{1000}\frac{1}{\sqrt{10}}\cdot \left( \frac{1}{5}\right)^i
\end{align*}
and by induction $\expec{G(x_i)} \le \frac{1}{5^i} \frac{1}{\sqrt{10}}$. We halt when $(\frac{1}{5})^i = O(\sqrt{\epsilon/\gap})$ and hence $c_1(i) = O(\sqrt{\epsilon/\gap})$.

In order to apply Theorem~\ref{thm:powermethod-perturb} we need a subroutine $\rayquoth{x}$ that lets us approximate $\rayquot(x)$ to within an additive error $\frac{1}{30} (\lambda-\lambda_1) = O(\gap \cdot\lambda_1)$. Theorem~\ref{online_rayleigh_estimation} gives us such a routine, requiring $O \left (\frac{\nvar(\dist)\log \log(\gap/\epsilon)}{\gap^2} \right) = O(\frac{\nvar(\dist)}{\gap \cdot \epsilon})$ samples to succeed with probability $1- O\left(\frac{1}{\log(\gap/\epsilon)}\right)$ (since $\epsilon < \gap$). Union bounding, the estimation succeeds in all rounds with constant probability.

By Corollary \ref{streaming_solver} with $c = \Theta(c_1(i)^2)$ the cost for solving each linear system solve is
$
O \left ( \frac{\nvar(\dist)}{\gap^2c_1(i)^2} \right )
$. Since $c_1(i)$ multiplies by a constant factor with each iteration the cost over all $O(\log(\gap/d\epsilon)$ iterations is just a truncated geometric series and is proportional to cost in the last iteration, when $c = \Theta\left(\frac{ \epsilon}{\gap}\right)$. So the total cost for solving the linear systems is 
$
O \left ( \frac{\nvar(\dist)}{\gap \cdot \epsilon} \right )
$. Adding this to the number of samples for the Rayleigh quotient estimation yields the result.
\end{proof}

\section{Parameter Estimation for Offline Eigenvector Computation}\label{parameter_free}


In Section \ref{sec:offline}, in order to invoke Theorems \ref{thm:powermethod-perturb} and \ref{thm:init-offline} we assumed knowledge of some $\lambda$ with $(1 + c_1 \cdot \gap)\lambda_1 \le \lambda \le (1 + c_2 \cdot \gap)\lambda_1$ for some small constant $c_1$ and $c_2$. Here we show how to estimate this parameter using Algorithm \ref{algo:eigestimate}, incurring a modest additional runtime cost.

In this section, for simplicity we initially assume that we have oracle access to compute $\mb_{\lambda}^{-1}x$ for any given $x$, and any $\lambda > \lambda_1$. We will then show how to achieve the same results when we can only compute $\mb_{\lambda}^{-1}x$ approximately. 
We use a result of~\cite{Musco2015} that gives gap free bounds for computing eigenvalues using the power method. The following is a specialization of Theorem 1 from~\cite{Musco2015}:
\begin{theorem}\label{thm:musco}
	For any $\epsilon > 0$, any matrix $\mM \in \R^{d\times d}$ with eigenvalues $\lambda_1,...,\lambda_d$, and $k \leq d$, let $\bv{W} \in \mathbb{R}^{d \times k}$ be a matrix with entries drawn independently from $\mathcal{N}(0,1)$. Let $eigEstimate (\bv{Y})$ be a function returning for each $i$, $\tilde \lambda_i = \tilde v_i^\top \bv{M} \tilde v_i$ where $\tilde v_i$ is the $i^{th}$ largest left singular vector of  $\bv{Y}$. Then setting $[\tilde \lambda_1,...,\tilde \lambda_k ] = eigEstimate\left (\bv{M}^t \bv{W} \right )$, for some fixed constant $c$ and $t = c\alpha \log d$ for any $\alpha > 1$, with probability $1-\frac{1}{d^{10}}$, we have for all $i$:
	\begin{align*}
	|\tilde \lambda_i - \lambda_i|  \le \frac{1}{\alpha}\lambda_{k+1}
	\end{align*}
\end{theorem}

\begin{algorithm}[t]
	\caption{Estimating the eigenvalue and the eigengap}
	\begin{algorithmic}[1]
		\renewcommand{\algorithmicrequire}{\textbf{Input: }}
		\renewcommand{\algorithmicensure}{\textbf{Output: }}
		\REQUIRE $\bv{A}\in\mathbb{R}^{n\times d},\;\alpha$
		\STATE $w = \left[w_{1},w_{2}\right]\leftarrow\mathcal{N}\left(0,1\right)^{d\times 2}$
		\STATE $t\leftarrow O\left(\alpha\log d\right)$
		\STATE $\left[\lamtilij 01,\lamtilij 02\right]\leftarrow eigEstimate\left(\left(\bv{A}^{T}\bv{A}\right)^{t}w\right)$
		\STATE $\lambari 0\leftarrow(1+\frac{1}{2})\lamtilij 01$
		\STATE $i\leftarrow0$
		\WHILE{$\lambari i-\lamtilij i1<\frac{1}{10}\left(\lambari i-\lamtilij i2\right)$}
		\STATE $i\leftarrow i+1$
		\STATE $w = \left[w_{1},w_{2}\right]\leftarrow\mathcal{N}\left(0,1\right)^{d\times 2}$
		\STATE $\left[\lamhatij i1,\lamhatij i2\right]\leftarrow eigEstimate\left(\left(\lambari{i-1}\bv{I}-\bv{A}^{T}\bv{A}\right)^{-t}w\right)$
		\STATE $\left[\lamtilij i1,\lamtilij i2\right]\leftarrow\left[\lambari{i-1}-\frac{1}{\lamhatij i1},\lambari{i-1}-\frac{1}{\lamhatij i2}\right]$
		\STATE $\lambari i\leftarrow\frac{1}{2}\left(\lamtilij i1+\lambari{i-1}\right)$
		\ENDWHILE
		\ENSURE $\lambda$
	\end{algorithmic}
	\label{algo:eigestimate}
\end{algorithm}

Throughout the proof, we assume $\alpha$ is picked to be some large constant - e.g. $\alpha>100$. Theorem \ref{thm:musco} implies:
\begin{lemma}
	\label{lem:topeig1} Conditioning on the event that Theorem \ref{thm:musco} holds for all iterates $i$, then the iterates
	of Algorithm \ref{algo:eigestimate} satisfy:
	\begin{align*}
	0\leq\lamj 1-\lamtilij 01\leq\frac{1}{\alpha}\lamj 1\;\;\mbox{and}\;\; & \frac{1}{2}\left(1-\frac{3}{\alpha}\right)\lamj 1\leq\lambari 0-\lamj 1\leq\frac{1}{2}\lamj 1,\;\;\mbox{and,}
	\end{align*}
	\begin{align*}
	0\leq\lamj 1-\lamtilij i1\leq\frac{1}{\alpha-1}\left(\lambari{i-1}-\lamj 1\right)\;\;\mbox{and}\;\; & \frac{1}{2}\left(1-\frac{1}{\alpha-1}\right)\left(\lambari{i-1}-\lamj 1\right)\leq\lambari i-\lamj 1\leq\frac{1}{2}\left(\lambari{i-1}-\lamj 1\right).
	\end{align*}
\end{lemma}
\begin{proof}
	The proof can be decomposed into two parts:
	
	\textbf{Part I (Lines 3-4):} Theorem \ref{thm:musco} tells us that $\lamtilij 01\geq\left(1-\frac{1}{\alpha}\right)\lamj 1$.
	This means that we have
	\begin{align*}
	0\leq\lamj 1-\lamtilij 01\leq\frac{1}{\alpha}\lamj 1\;\;\mbox{and}\;\; & \frac{1}{2}\left(1-\frac{3}{\alpha}\right)\lamj 1\leq\lambari 0-\lamj 1\leq\frac{1}{2}\lamj 1.
	\end{align*}

	\textbf{Part II (Lines 5-6):} Consider now iteration $i$. We now
	apply Theorem \ref{thm:musco} to the matrix $\left(\lambari{i-1}\bv{I}-\bv{A}^{T}\bv{A}\right)^{-1}$.
	The top eigenvalue of this matrix is $\left(\lambari{i-1}-\lamj 1\right)^{-1}$.
	This means that we have $\left(1-\frac{1}{\alpha}\right)\left(\lambari{i-1}-\lamj 1\right)^{-1}\leq\lamhatij i1\leq\left(\lambari{i-1}-\lamj 1\right)^{-1}$,
	and hence we have, 
	\begin{align*}
	0\leq\lamj 1-\lamtilij i1\leq\frac{1}{\alpha-1}\left(\lambari{i-1}-\lamj 1\right)\;\;\mbox{and}\;\; & \frac{1}{2}\left(1-\frac{1}{\alpha-1}\right)\left(\lambari{i-1}-\lamj 1\right)\leq\lambari i-\lamj 1\leq\frac{1}{2}\left(\lambari{i-1}-\lamj 1\right).
	\end{align*}

	This proves the lemma.\end{proof}
\begin{lemma}
	\label{lem:secondeig}Recall we denote $\lamj 2\defeq\lamj 2\left(\bv{A}^{T}\bv{A}\right)$
	and $\gap\defeq\frac{\lamj 1-\lamj 2}{\lamj 1}$. Then conditioning on the event that Theorem \ref{thm:musco} holds for all iterates $i$, the iterates
	of Algorithm \ref{algo:eigestimate} satisfy $\abs{\lamj 2-\lamtilij i2}\leq\frac{1}{\alpha-1}\left(\lambari{i-1}-\lamj 2\right)$,
	and $\lambari i-\lamtilij i2\geq\frac{\gap \lamj 1}{4}$.\end{lemma}
\begin{proof}
	Since $\left(\lambari{i-1}-\lamj 2\right)^{-1}$ is the second eigenvalue
	of the matrix $\left(\lambari{i-1}\bv{I}-\bv{A}^{T}\bv{A}\right)^{-1}$, Theorem \ref{thm:musco} 
	tells us that
	\[
	\left(1-\frac{1}{\alpha}\right)\left(\lambari{i-1}-\lamj 2\right)^{-1}\leq\lamhatij i2\leq\left(1+\frac{1}{\alpha}\right)\left(\lambari{i-1}-\lamj 2\right)^{-1}.
	\]

	This immediately yields the first claim. For the second claim, we
	notice that
	\begin{align*}
	\lambari i-\lamtilij i2 
	&=\lambari i-\lamj 2+\lamj 2-\lamtilij i2\\
	& \stackrel{(\zeta_{1})}{\geq}\lambari i-\lamj 2-\frac{1}{\alpha-1}\left(\lambari{i-1}-\lamj 2\right) \\
	& =\lambari i-\lamj 1-\frac{1}{\alpha-1}\left(\lambari{i-1}-\lamj 1\right)+\left(1-\frac{1}{\alpha-1}\right)\left(\lamj 1-\lamj 2\right) \\
	& \stackrel{\left(\zeta_{2}\right)}{\geq}\frac{1}{2}\left(1-\frac{3}{\alpha-1}\right)\left(\lambari{i-1}-\lamj 1\right)+\left(1-\frac{1}{\alpha-1}\right)\left(\lamj 1-\lamj 2\right)\geq\frac{\gap\lamj 1}{4},
	\end{align*}

	where $\left(\zeta_{1}\right)$ follows from the first claim of this lemma,
	and $\left(\zeta_{2}\right)$ follows from Lemma \ref{lem:topeig1}.
\end{proof}

We now state and prove the main result in this section:
\begin{theorem}\label{thm:parafree}
Suppose $\alpha>100$, and after $T$ iterations, Algorithm \ref{algo:eigestimate} exits. Then
with probability $1- \frac{\ceil{\log\frac{10}{\gap}}+1}{d^{10}}$, we have
$T \le \ceil{\log\frac{10}{\gap}}+1$, and:
\begin{equation*}
\left(1+\frac{\gap}{120}\right)\lambda_1 
\le \lambari T 
\le \left(1+\frac{\gap}{8}\right)\lambda_1
\end{equation*}
\end{theorem}

\begin{proof}
By union bound, we know with probability $1- \frac{\ceil{\log\frac{10}{\gap}}+1}{d^{10}}$, Theorem \ref{thm:musco} will hold for all iterates where $i\le\ceil{\log\frac{10}{\gap}}+1$.

Let $\overline{i}=\ceil{\log\frac{10}{\gap}}$, suppose the algorithm has not exited yet after $\overline{i}$
iterations, then since $\lambari i-\lamj 1$ decays geometrically, we have $\lambari{\overline{i}}-\lamj 1\leq\frac{\gap\lamj 1}{10}$.
Therefore, Lemmas \ref{lem:topeig1} and \ref{lem:secondeig} imply
that $\lambari{\overline{i}+1}-\lamtilij{\overline{i}+1}1\leq\left(\frac{1}{2}+\frac{1}{\alpha-1}\right)\left(\lambari{\overline{i}}-\lamj 1\right)\leq\frac{\gap\lamj 1}{15}$,
and 
\begin{align*}
\lambari{\overline{i}+1}-\lamtilij{\overline{i}+1}2
&\geq\lambari{\overline{i}+1}-\lamj 2-\abs{\lamj 2-\lamtilij{\overline{i}+1}2}\geq\lamj 1-\lamj 2-\frac{1}{\alpha-1}\left(\lambari{\overline{i}}-\lamj 2\right)\\
&=\gap\lamj 1-\frac{1}{\alpha-1}\left(\lambari{\overline{i}}-\lamj 1+\lamj 1-\lamj 2\right)\geq\frac{3}{4}\gap\lamj 1
\end{align*}
This means that the exit condition on Line $6$ must be triggered in $\overline{i}+1$ iteration, proving
the first part of the lemma.

For upper bound, by Lemmas \ref{lem:topeig1}, \ref{lem:secondeig} and exit condition we know:
\begin{align*}
\lambari T- \lambda_1 &\le \lambari T-\lamtilij T1
\le \frac{1}{10} (\lambari T-\lamtilij T2)
\le \frac{1}{10} \left(\lambari T - \lambda_2 + \abs{\lamj 2-\lamtilij T2}\right)\\
&\leq \frac{1}{10} \left(\lambari T - \lambda_2 + \frac{1}{\alpha-1}(\lambari{T-1}-\lamj 2)\right)\\
&= \frac{1}{10} \left( \frac{\alpha}{\alpha-1} \gap \lambda_1 + (\lambari T  - \lambda_1) +  \frac{1}{\alpha-1}\left(\lambari{T-1}-\lamj 1\right) \right)\\
&\le \frac{1}{10} \left( \frac{\alpha}{\alpha-1} \gap \lambda_1 + \frac{\alpha}{\alpha-2}\left(\lambari{T}-\lamj 1\right) \right)
\end{align*}
Since $\alpha > 100$, this directly implies $\lambari{T}-\lamj 1 \le \frac{\gap}{8}\lambda_1$.

For lower bound, since as long as the Algorithm \ref{algo:eigestimate} does not exists, by Lemmas \ref{lem:secondeig}, we have
$\lambari {T-1}-\lamtilij {T-1}1\geq\frac{1}{10}\left(\lambari {T-1}-\lamtilij {T-1}2\right)\geq\frac{\gap\lamj 1}{40}$, and thus:
\begin{align*}
\lambari {T-1}- \lambda_1 &= \lambari {T-1}-\lamtilij {T-1}1 - (\lambda_1 - \lamtilij {T-1}1)
\ge \frac{\gap\lamj 1}{40} - \frac{1}{\alpha-1}\left(\lambari{T-1}-\lamj 1\right) \\
& \ge \frac{\gap\lamj 1}{40} - \frac{2}{\alpha-2}\left(\lambari{T}-\lamj 1\right)
\ge \frac{\gap\lamj 1}{50}
\end{align*}
By Lemma \ref{lem:topeig1}, we know $\lambari {T}- \lambda_1\ge \frac{1}{2}(1-\frac{1}{\alpha-1}(\lambari {T-1}- \lambda_1)) >\frac{\gap}{120}\lamj 1$
\end{proof}

Note that, although we proved the upper bound and lower bound in Theorem \ref{thm:parafree} with specific constants coefficient $\frac{1}{8}$ and $\frac{1}{120}$, this analysis can easily be extended  to any smaller constants by modifying the constant in the exit condition, and choosing $\alpha$ larger.
Also in the failure probability 
$$1- \frac{\ceil{\log\frac{10}{\gap}}+1}{d^{10}},$$ the term $d^{10}$ can be replaced by any $\poly(d)$ by adjusting the constant in setting $t \leftarrow O(\alpha \log d)$ in Algorithm \ref{algo:eigestimate}. Assuming $\log \frac{1}{\gap} < \poly(d)$, thus gives that Theorem \ref{thm:parafree} returns a correct result with high probability.

Finally, we can also bound the runtime of algorithm \ref{algo:eigestimate}, when we use SVRG based approximate linear system solvers for $\bv{B}_\lambda$.
\begin{theorem}\label{thm:parafree_runtime}
With probability $1- O(\frac{1}{d^{10}}\log\frac{1}{\gap})$, Algorithm \ref{algo:eigestimate} runs in time $$O\left (\left [\nnz(\bv A) + \frac{d\nrank(\bv A)}{\gap^2} \right ] \cdot \log^3\left (\frac{d}{\gap} \right ) \right)$$.
\end{theorem}

\begin{proof}
By Theorem \ref{thm:parafree}, we know only $O(\log 1/\gap)$ iterations of the algorithm are needed.
In each iteration, the runtime is dominated by running $eigEstimate\left(\left(\lambari{i-1}\bv{I}-\bv{A}^{T}\bv{A}\right)^{-t}w\right)$, which is dominated by computing $\left(\lambari{i-1}\bv{I}-\bv{A}^{T}\bv{A}\right)^{-t}w$. 
Since $t=O(\log d)$, it's easy to verify that: to make Theorem \ref{thm:musco} hold, we only need to approximate $\left(\lambari{i-1}\bv{I}-\bv{A}^{T}\bv{A}\right)^{-1}w$ up to accuracy $\poly(\gap/d)$. By Theorem \ref{offline_solver}, we know this approximation can be calculated in time $$O\left (\left [ \nnz(\bv A) + \frac{d\nrank(\bv A)\lambda_1^2}{(\lambari{i-1} - \lambda_1)^2} \right ] \cdot \log\left (\frac{d}{\gap}\right ) \right)$$. Combining Theorem \ref{thm:parafree} with Lemma \ref{lem:topeig1}, we know $\lambari{i-1} - \lambda_1 \ge 
\lambari{T} - \lambda_1 \ge \frac{\gap}{120}$, thus approximately solving $\left(\lambari{i-1}\bv{I}-\bv{A}^{T}\bv{A}\right)^{-1}w$ can be done in time $\tilde{O}\left (\nnz(\bv A) + \frac{d\nrank(\bv A)}{\gap^2} \right)$. Finally, since the runtime of Algorithm \ref{algo:eigestimate} is dominated by repeating this subroutine $t\times T = O(\log d \cdot \log (1/\gap))$ times, we finish the proof.
\end{proof}

Note that we can accelerate the runtime of Algorithm \ref{algo:eigestimate} to $\tilde O \left (\frac{\nnz(\bv A)^{3/4}(d\nrank(\bv A))^{1/4}}{\sqrt{\gap}} \right )$, by simply
replacing the base solver for $\left(\lambari{i-1}\bv{I}-\bv{A}^{T}\bv{A}\right)^{-1}w$ with the accelerated solver in Theorem \ref{accelerated_offline_solver}.

\section{Lower Bounds}\label{sec:lower}

Here we show that our online eigenvector estimation algorithm (Theorem \ref{warmstart_online_theorem})  is asymptotically optimal - as sample size grows large it achieves optimal accuracy as a function of sample size. 
We rely on the following lower bound for eigenvector estimation in the Gaussian spike model:

\begin{lemma}[Lower bound for Gaussian Spike Model \cite{birnbaum2013minimax}] \label{lowerbound_spike_lemma}
Suppose data is generated as 
\begin{equation}\label{spike_model}
a_i = \sqrt{\lambda}\iota_i{v^\star} + Z_i
\end{equation}
where $\iota_i \sim \mathcal{N}(0, 1)$, and $Z_i \sim \mathcal{N}(0, I_d)$. 
Let $\hat{v}$ be some estimator of the top eigenvector $v^\star$. 
Then, there is some universal constant $c_0$, so that for $n$ sufficiently large, we have:
\begin{equation*}
\inf_{\hat{v}} \max_{v^\star \in \mathbb{S}^{d-1}} \E \norm{\hat{v} - v^\star}_2
\ge c_0 \frac{(1+\lambda)d}{\lambda^2 n}
\end{equation*}
\end{lemma}

\begin{theorem}\label{lowerbound_online_theorem}
Consider the problem of estimating the top eigenvector $v_1$ of $\E_{a\sim \mathcal{D}} aa^\top$, where we observe $n$ i.i.d samples from unknown distribution $\mathcal{D}$. If $\gap < 0.9$, then there exists some universal constant c, such that for any estimator $\hat{v}$ of top eigenvector, there always exists some hard distribution $\mathcal{D}$ so that for $n$ sufficiently large:
\begin{align*}
\E\norm{\hat{v}-v_1}^2_2 \ge c  \frac{\nvar(\dist)}{\gap^2n} 
\end{align*}
\end{theorem}

\begin{proof}
Suppose the claim of theorem is not true, then there exist some estimator $\hat{v}$ so that 
\begin{equation*}
\E\norm{\hat{v}-v_1}^2_2 < c'  \frac{\nvar(\dist)}{\gap^2n} 
\end{equation*}
holds for all distribution $\mathcal{D}$, and for any fixed constant $c'$ when $n$ is sufficiently large.

Let distribution $\mathcal{D}$ be the Gaussian Spike Model specified by Eq.(\ref{spike_model}), then
by calculation, it's not hard to verify that:
\begin{equation*}
\nvar(\dist) = \frac{\norm{\E_{a \sim \dist}\left [ \left (a a^\top \right )^2 \right ]}_2}{\norm{\E_{a \sim \dist} (aa^\top)}_2^2} =\frac{d+2+3\lambda}{1+\lambda}
\end{equation*}Since we know $\gap = \frac{\lambda}{1+\lambda} <0.9$, this implies $\lambda <9$, which gives
$\nvar(\dist) < \frac{d+29}{1+\lambda} <\frac{30d}{1+\lambda}$.
Therefore, we have that:
\begin{equation*}
\E\norm{\hat{v}-v^\star}^2_2 < c'  \frac{\nvar(\dist)}{\gap^2n} <30c' \frac{(1+\lambda)d}{\lambda^2 n}
\end{equation*}
holds for all $v^\star \in \mathbb{S}^{d-1}$. Choose $c' = \frac{c_0}{30}$ in Lemma \ref{lowerbound_spike_lemma} we have a contradiction.
\end{proof}

$\norm{\hat{v}-v_1}^2_2  = 2 - 2  \hat{v}^\top v_1$, so this bound implies that- to obtain $|\hat{v}^\top v_1| \ge 1- \epsilon$,  we need $\frac{\nvar(\dist)}{\gap^2 n} = O(\epsilon)$ so $n = \Theta \left (\frac{\nvar(\dist)}{\gap^2 \epsilon} \right )$. This exactly matches the sample complexity given by Theorem \ref{warmstart_online_theorem}.
\section{Gap-Free Bounds}\label{sec:gapfree}

In this section we demonstrate that our techniques can easily be extended to obtain gap-free runtime bounds, for the regime when $\epsilon \ge \gap$. In many ways these bounds are actually much easier to achieve than the gap dependent bounds since they require less careful error analysis.

Let $\epsilon$ be our error parameter and $m$ be the number of eigenvalues of $\Sigma$ that are $\ge (1-\epsilon/2)\lambda_1$.
Choose $\lambda = \lambda_1 + \epsilon/100$. We have $\lambda_1(\mb^{-1}) = \frac{100}{\epsilon\lambda_1}$. For $i > m$ we have $\lambda_i(\mb^{-1}) < \frac{2}{\epsilon\lambda_1}$. $\kappa(\bv{B}^{-1}) \le \frac{100}{\epsilon}$. 

Let $\bv{V}_b$ have columns equal to all \emph{bottom} eigenvectors with eigenvalues $\lambda_i < (1-\epsilon/2) \lambda_1$. Let $\bv{V}_t$ have columns equal to  the $m$ remaining \emph{top} eigenvectors.
We define a simple modified potential:
\begin{align*}
\bar G(x) \eqdef \frac{\norm{\bv{P}_{\bv{V}_b}x}_\mb}{\norm{\bv{P}_{v_1}x}_\mb} = \frac{\sqrt{\sum_{i > m} \frac{\alpha_i^2}{\lambda_i(\mb^{-1})}}}{\sqrt{\frac{\alpha_1^2}{\lambda_1(\mb^{-1})}}}
\end{align*}

We have the following Lemma connecting this potential function to eigenvalue error:
\begin{lemma}
For unit $x$, if $\bar G(x) \le c\sqrt{\epsilon}$ for sufficiently small constant $c$ then $\lambda_1 - x^\top \bv{\Sigma}x \le \epsilon \lambda_1$.
\end{lemma} 
\begin{proof}
\begin{align*}
\bar G(x) \ge \frac{\norm{\bv{P}_{\bv{V}_b}x}_2}{\norm{\bv{P}_{v_1}x}_2}  \ge \frac{\norm{\bv{P}_{\bv{V}_b}x}_2}{\norm{\bv{P}_{\bv{V}_t}x}_2}  
\end{align*}
So if $\bar G(x) \le c\sqrt{\epsilon}$ then $\norm{\bv{P}_{\bv{V}_t}x}^2_2 c^2\epsilon \ge \norm{\bv{P}_{\bv{V}_b}x}^2_2$ and since $\norm{\bv{P}_{\bv{V}_t}x}^2_2 + \norm{\bv{P}_{\bv{V}_b}x}^2_2 = 1$, this gives $\norm{\bv{P}_{\bv{V}_t}x}^2_2 \ge \frac{1}{1+c^2\epsilon}$. So we have $x^T \bv \Sigma x \ge \bv{P}_{\bv{V}_t}x^T \bv \Sigma x\bv{P}_{\bv{V}_t} \ge \frac{(1-\epsilon/2)\lambda_1}{1+c^2\epsilon} \ge 1-\epsilon$ for small enough $c$, giving the lemma.
\end{proof}

We now follow the proof of Lemma \ref{thm:init-offline}, which is actually much simpler in the gap-free case.

\begin{theorem}[Approximate Shifted-and-Inverted Power Method -- Gap-Free]\label{burnInGapFree}
	Suppose we randomly initialize $x_0$ as in Lemma \ref{lem:init-random} and suppose we have access to a subroutine $\solve{\cdot}$ such that
	\begin{align*}
		\expec{\norm{\solve{x}-\mb^{-1}x}_{\mb}} \leq \frac{\epsilon^3}{3000 d^{21}} \sqrt{\lambda_d(\mb^{-1})}
	\end{align*}
	Then the following procedure,
	\begin{align*}
		x_{t} = \solve{x_{t-1}}/\norm{\solve{x_{t-1}}}
	\end{align*}
	after $T = O\left(\log d/\epsilon \right)$ iterations  satisfies:
	\begin{align*}
		\bar G(x_T) \leq c\sqrt{\epsilon},
	\end{align*}
	with probability greater than $1- O(\frac{1}{d^{10}})$.
\end{theorem}

\begin{proof}

By Lemma \ref{lem:init-random}, we know with at least probability $1- O(\frac{1}{d^{10}})$, 
we have $\bar G(x_0) \le G(x_0) \le \sqrt{\kappa(\mb^{-1})} d^{10.5} = \frac{100d^{10.5}}{\epsilon}$.
We want to show by induction that at iteration $i$ we have $\bar G(x_i) \le \frac{1}{2^i} \cdot \frac{100d^{10.5}}{\epsilon}$, which will give us the lemma if we set $T = \log_2 \left ( \frac{100d^{10.5}}{c\epsilon^{1.5}} \right) = O(\log(d/\epsilon))$.

Let $\xhat = \solve{x}$ and $\xi = \xhat - \mb^{-1}x$. Following Lemma \ref{thm:init-offline} we have:
\begin{align*}
\norm{\bv{P}_{\bv{V}_b}\left(\xhat \right)}_{\mb} &\leq \norm{\bv{P}_{\bv{V}_b}\left(\mb^{-1} x \right)}_{\mb} + \norm{\bv{P}_{\bv{V}_b}\left(\xi\right)}_\mb \leq \norm{\bv{P}_{\bv{V}_b}\left(\mb^{-1} x \right)}_{\mb} + \norm{\xi}_{\mb} \\
&=\sqrt{\sum_{i>m}\alpha_i^{2}\lambda_{i}(\mb^{-1})}+\norm{\xi}_{\mb}\\
&\le \lambda_{m+1}(\mb^{-1}) \left (\sqrt{\sum_{i>m} \frac{\alpha_i^2}{\lambda_i(\mb^{-1})}} + \frac{\epsilon^3}{3000d^{21}\sqrt{\lambda_{m+1}(\mb^{-1})}}\right )\\
&\le 2\lambda_{m+1}(\mb^{-1}) \max \left \{ \sqrt{\sum_{i>m} \frac{\alpha_i^2}{\lambda_i(\mb^{-1})}},\frac{\epsilon^3}{3000d^{21}\sqrt{\lambda_{m+1}(\mb^{-1})}}\right \}
\end{align*}
and 
\begin{align*}
\norm{\bv{P}_{v_1}\left(\xhat \right)}_{\mb}
	&\geq  \norm{\bv{P}_{v_1}\left(\mb^{-1} x \right)}_{\mb} - \norm{\bv{P}_{v_1}\left(\xi\right)}_\mb
	\geq \norm{\bv{P}_{v_1}\left(\mb^{-1} x \right)}_{\mb} - \norm{\xi}_{\mb} \\
	&=\sqrt{\alpha_1^2 \lambda_1(\bv{B}^{-1})} - \norm{\xi}_{\mb}\\
	& \ge \lambda_1(\bv{B}^{-1}) \sqrt{\frac{\alpha_1^2 - \frac{\epsilon^6}{(3000d^{21})^2}}{ \lambda_1(\bv{B}^{-1})}}.
\end{align*}

Initially, we have with high probability, by the argument in Lemma \ref{lem:init-random}, $\alpha_1 \ge \frac{1}{d^{10}}$ so we have $\norm{\bv{P}_{v_1}\left(\xhat \right)}_{\mb} \ge \frac{\lambda_1(\bv{B}^{-1})}{2} \sqrt{\frac{\alpha_1^2}{ \lambda_1(\bv{B}^{-1})}}$. This also holds by induction in each iteration.

Let $\hat \alpha_1=|v_1^\top \xhat|/\norm{\xhat}_2$. $\norm{\bv{P}_{v_1}\left(\xhat \right)}_{\mb}^2 = \frac{\hat \alpha_1^2\norm{\hat x}_2^2}{\lambda_1(\bv{B}^{-1})}$ so we have 
\begin{align*}
\hat \alpha_1^2 &\ge \frac{\lambda_1(\bv{B}^{-1})^2}{\norm{\hat x}_2^2} \left (\alpha_1^2 - \frac{\epsilon^6}{(3000d^{21})^2}\right )
\end{align*}
and since $\norm{\hat x}_2^2 \le 2 \left (\norm{\bv{B}^{-1}x}_2^2 + 2\norm{\xi}_2^2 \right ) \le \lambda_1(\bv{B}^{-1})^2 + 2\frac{\epsilon^6}{(3000d^{21})^2} \le \lambda_1(\bv{B}^{-1})^2\left (2 + 2\frac{\epsilon^6}{(3000d^{21})^2} \right )$ we have:
\begin{align*}
\hat \alpha_1^2 &\ge \frac{1}{2.1} \left (\alpha_1^2 - \frac{\epsilon^6}{(3000d^{21})^2}\right ) \ge \frac{1}{3} \alpha_1^2.
\end{align*}

So over all $\log_2 \left ( \frac{100d^{10.5}}{c\epsilon^{1.5}} \right)$ iterations, we always have $\hat \alpha_1^2 \ge \frac{1}{d^{10}}\cdot  \left ( \frac{c\epsilon^{1.5}}{100d^{10.5}}\right )^{\log_2 3}$ and so $\frac{\epsilon^6}{(3000d^{21})^2}  << 1/2\alpha_1^2$.
Combining the above bounds:

\begin{align*}
	\bar G(\xhat)  &\le \frac{2\lamiBinv{m+1}}{\lamiBinv{1}/2} \cdot \frac{\max \left \{ \sqrt{\sum_{i>m} \frac{\alpha_i^2}{\lambda_i(\mb^{-1})}},\frac{\epsilon^3}{3000d^{21}\sqrt{\lambda_{m+1}(\mb^{-1})}}\right \}}{\sqrt{\frac{\alpha_1^2}{ \lambda_1(\bv{B}^{-1})}}}\\
	&\le \frac{4}{50} \max \left \{\bar G(x), O(\sqrt{\epsilon}) \right \}.
\end{align*}
This is enough to give the Theorem.
\end{proof}

Finally, we combine Theorem \ref{burnInGapFree} with the SVRG based solvers of  Theorem \ref{offline_solver} and \ref{accelerated_offline_solver} to obtain:

\begin{theorem}[Gap-Free Shifted-and-Inverted Power Method With SVRG]\label{main_gapfree_theorem}
Let $\bv{B} = \lambda \bv{I} - \bv{A}^\top \bv{A}$ for $\lambda = \left ( 1+\frac{\epsilon}{100}\right)$ and let $x_0 \sim \mathcal{N}(0,\bv I)$ be a random initial vector. Running the inverted power method on $\bv{B}$ initialized with $x_0$, using the SVRG solver from Theorem \ref{offline_solver} to approximately apply $\bv{B}^{-1}$ at each step, returns $x$ such that with probability $1-O\left (\frac{1}{d^{10}}\right)$, $x^\top \bv{\Sigma}x \ge (1-\epsilon) \lambda_1$ in time $$O \left (\left(\nnz(\bv A) + \frac{d \nrank(\bv A)}{\epsilon^2} \right )\cdot \log^2\left(\frac{d}{\epsilon}\right) \right ).$$
\end{theorem}

\begin{theorem}[Accelerated Gap-Free Shifted-and-Inverted Power Method With SVRG]\label{acell_gapfree_theorem}
Let $\bv{B} = \lambda \bv{I} - \bv{A}^\top \bv{A}$ for $\lambda = \left ( 1+\frac{\epsilon}{100}\right)$ and let $x_0 \sim \mathcal{N}(0,\bv I)$ be a random initial vector. Running the inverted power method on $\bv{B}$ initialized with $x_0$, using the SVRG solver from Theorem \ref{accelerated_offline_solver} to approximately apply $\bv{B}^{-1}$ at each step, returns $x$ such that with probability $1-O\left (\frac{1}{d^{10}}\right)$, $x^\top \bv{\Sigma}x \ge (1-\epsilon) \lambda_1$ in total time $$O \left (\frac{\nnz(\bv A)^{3/4} (d \nrank(\bv A))^{1/4}}{\sqrt{\epsilon}} \cdot \log^3\left(\frac{d}{\epsilon}\right) \right ).$$
\end{theorem}

\section{Acknowledgements}
Sham Kakade acknowledges funding from the Washington Research
Foundation for innovation in Data-intensive Discovery.

\bibliographystyle{alpha}
\bibliography{refs}

\appendix

\section{Appendix}

\begin{lemma}[Eigenvector Estimation via Spectral Norm Matrix Approximation]\label{spectral_error_conversion}
Let $\bv{A}^\top\bv{A}$ have top eigenvector $1$, top eigenvector $v_1$ and eigenvalue gap $\gap$. Let $\bv{B}^\top \bv{B}$ be some matrix with $\norm{\bv{A}^\top\bv{A}-\bv{B}^\top\bv{B}}_2 \le O(\sqrt{\epsilon} \cdot \gap)$. Let $x$ be the top eigenvector of $\bv{B}^\top \bv{B}$. Then:
\begin{align*}
|x^\top v_1 | \ge 1 - \epsilon ~.
\end{align*}
\end{lemma}
\begin{proof}
We can any unit vector $y$ as $y = c_1 v_1 + c_2 v_2$ where $v_2$ is the component of $x$ orthogonal to $v_1$ and $c_1^2 + c_2^2 = 1$. We know that
\begin{align*}
v_1^\top \bv{B}^\top \bv{B} v_1 &= v_1^\top \bv{A}^\top\bv{A}v_1-v_1^T(\bv{A}^\top\bv{A}-\bv{B}^\top\bv{B})v_1\\
1 -\sqrt{\epsilon} \gap &\le v_1^\top \bv{B}^\top \bv{B} v_1\le 1 +\sqrt{\epsilon} \gap
\end{align*}
Similarly we can compute:
\begin{align*}
v_2^\top \bv{B}^\top \bv{B}  v_2 &= v_2^\top \bv{A}^\top\bv{A}v_2-v_2^T(\bv{A}^\top\bv{A}-\bv{B}^\top\bv{B})v_2\\
1 -\gap -\sqrt{\epsilon} \gap &\le v_2^\top \bv{B}^\top \bv{B} v_2 \le 1 -\gap +\sqrt{\epsilon} \gap.
\end{align*}
and 
\begin{align*}
|v_1^\top \bv{B}^\top \bv{B}  v_2| &= |v_1^\top \bv{A}^\top\bv{A}v_2-v_1^T(\bv{A}^\top\bv{A}-\bv{B}^\top\bv{B})v_2|\\
&\le \sqrt{\epsilon} \gap.
\end{align*}

We have $x^\top \bv{B}\bv{B}^\top x = c_1^2 (v_1^\top \bv{B}^\top \bv{B} v_1) + c_2^2 (v_2^\top \bv{B}^\top \bv{B} v_2) + 2c_1c_2 \cdot v_2^\top \bv{B}^\top \bv{B} v_1$. 

We want to bound $c_1 \ge 1-\epsilon$ so $c_1^2 \ge 1- O(\epsilon)$. Since $x$ is the top eigenvector of $\bv{BB}^\top$ we have:
\begin{align*}
x^\top \bv{B}\bv{B}^\top x &\ge v_1^\top \bv{B}\bv{B}^\top v_1\\
c_2^2 (v_2^\top \bv{B}^\top \bv{B} v_2) + 2c_2 v_2^\top \bv{B}^\top \bv{B} v_1 &\ge (1-c_1^2)v_1^\top \bv{B}\bv{B}^\top v_1\\
2\sqrt{1-c_1^2} \sqrt{\epsilon}\gap &\ge (1-c_1^2)\left (v_1^\top \bv{B}\bv{B}^\top v_1 -v_2^\top \bv{B}\bv{B}^\top v_2 \right )\\ 
\frac{1}{\sqrt{1-c_1^2}} &\ge \frac{(1-2\sqrt{\epsilon})\gap}{2\sqrt{\epsilon}\gap}\\
\frac{1}{1-c_1^2} &\ge \frac{1-5\sqrt{\epsilon}}{4\epsilon}
\end{align*}

This means we need have $1-c_1^2 \le O(\epsilon)$  meaning $c_1^2 \ge 1-O(\epsilon)$ as desired.
\end{proof}

\begin{lemma}[Inverted Power Method progress in $\ell_2$ and $\bv{B}$ norms]\label{same_progress}
	Let $x$ be a unit vector with $\inprod{x,v_1} \neq 0$ and let $\xtilde = \mb^{-1}w$, i.e. the power method update of $\mb^{-1}$ on $x$. Then, we have both:
	\begin{align}\label{lb_progress}
\frac{\norm{\bv{P}_{v_1^{\perp}}\xtilde}_{\mb}}{\norm{\bv{P}_{v_1}\xtilde}_{\mb}} \le \frac{\lambda_2(\bv{B}^{-1})}{\lambda_1(\bv{B}^{-1})} \cdot \frac{\norm{\bv{P}_{v_1^{\perp}}x}_{\mb}}{\norm{\bv{P}_{v_1}x}_{\mb}}
	\end{align}
	and 
	\begin{align}\label{l2_progress}
\frac{\norm{\bv{P}_{v_1^{\perp}}\xtilde}_{2}}{\norm{\bv{P}_{v_1}\xtilde}_{2}} \le \frac{\lambda_2(\bv{B}^{-1})}{\lambda_1(\bv{B}^{-1})} \cdot \frac{\norm{\bv{P}_{v_1^{\perp}}x}_{2}}{\norm{\bv{P}_{v_1}x}_{2}}
	\end{align}
\end{lemma}
\begin{proof}
\eqref{lb_progress} was already shown in Lemma \ref{thm:powermethod}. We show \eqref{l2_progress} similarly.

	Writing $x$ in the eigenbasis of $\bv{B}^{-1}$, we have $x=\sum_i \alpha_i v_i$ and  $\xtilde = \sum_i \alpha_i \lamiBinv{i}v_i$. Since $\inprod{x,v_1} \neq 0$, $\alpha_1 \neq 0$ and we have:
	\begin{align*}
		\frac{\norm{\bv{P}_{v_1^{\perp}}\xtilde}_{2}}{\norm{\bv{P}_{v_1}\xtilde}_{2}} = \frac{\sqrt{\sum_{i\geq 2} \alpha_i^2 \lambda^2_{i}(\bv{B}^{-1})}}{\sqrt{\alpha_1^2 \lambda^2_{1}(\bv{B}^{-1})}}
		\leq \frac{\lamiBinv{2}}{\lamiBinv{1}} \cdot \frac{\sqrt{\sum_{i\geq 2} \alpha_i^2}}{\sqrt{\alpha_1^2}}
		= \frac{\lamiBinv{2}}{\lamiBinv{1}} \cdot \frac{\norm{\bv{P}_{v_1^{\perp}}x}_{2}}{\norm{\bv{P}_{v_1}x}_{2}}.
	\end{align*}
\end{proof}

\end{document}